\documentclass[sigconf]{acmart}

\usepackage{balance}

\def\BibTeX{{\rm B\kern-.05em{\sc i\kern-.025em b}\kern-.08emT\kern-.1667em\lower.7ex\hbox{E}\kern-.125emX}}

\usepackage{booktabs} 
\usepackage{multirow}
\usepackage{graphicx}
\usepackage{balance}  
\usepackage{times}
\usepackage{algorithmicx}

\usepackage{algpseudocode}

\usepackage{appendix}

\usepackage{caption}
\usepackage{subfigure}

\usepackage[linesnumbered, ruled, vlined]{algorithm2e}

\SetKwRepeat{Do}{do}{while}

\usepackage{color}
\definecolor{Xiang}{rgb}{1,0,0}
\definecolor{Weilong}{rgb}{0,0,1}

\newcommand{\nop}[1]{}

\begin{document}

\fancyhead{}

%
\title{Efficient Join Processing Over Incomplete Data Streams (Technical Report)}

\author{Weilong Ren$^1$, Xiang Lian$^1$, and Kambiz Ghazinour $^1\ ^2$}
\affiliation{%
  \institution{$^1$ Department of Computer Science, Kent State University, Kent}
  \institution{$^2$ Center for Criminal Justice, Intelligence and Cybersecurity, State University of New York, Canton}
}
\email{{wren3, xlian}@kent.edu, ghazinourk@canton.edu}

\begin{abstract}
For decades, the join operator over fast data streams has always drawn much attention from the database community, due to its wide spectrum of real-world applications, such as online clustering, intrusion detection, sensor data monitoring, and so on. Existing works usually assume that the underlying streams to be joined are complete (without any missing values). However, this assumption may not always hold, since objects from streams may contain some missing attributes, due to various reasons such as packet losses, network congestion/failure, and so on. In this paper, we formalize an important problem, namely \textit{join over incomplete data streams} (Join-iDS), which retrieves joining object pairs from incomplete data streams with high confidences. We tackle the Join-iDS problem in the style of ``data imputation and query processing at the same time''. To enable this style, we design an effective and efficient cost-model-based imputation method via \textit{deferential dependency} (DD), devise effective pruning strategies to reduce the Join-iDS search space, and propose efficient algorithms via our proposed cost-model-based data synopsis/indexes. Extensive experiments have been conducted to verify the efficiency and effectiveness of our proposed Join-iDS approach on both real and synthetic data sets.
\end{abstract}

%
%
\nop{

\ccsdesc[500]{Computer systems organization~Embedded systems}
\ccsdesc[300]{Computer systems organization~Redundancy}
\ccsdesc{Computer systems organization~Robotics}
\ccsdesc[100]{Networks~Network reliability}
}

%
\copyrightyear{2019}
\acmYear{2019}
\acmConference[CIKM '19]{The 28th ACM International Conference on Information and
Knowledge Management}{November 3--7, 2019}{Beijing, China}
\acmBooktitle{The 28th ACM International Conference on Information and Knowledge
Management (CIKM '19), November 3--7, 2019, Beijing, China}
\acmPrice{15.00}
\acmDOI{10.1145/3357384.3357863}
\acmISBN{978-1-4503-6976-3/19/11}

%
\keywords{Join, Incomplete Data Streams, Join Over Incomplete Data Streams, Join-iDS}

%
\nop{
\begin{teaserfigure}
  \includegraphics[width=\textwidth]{sampleteaser}
  \caption{Seattle Mariners at Spring Training, 2010.}
  \Description{Enjoying the baseball game from the third-base seats. Ichiro Suzuki preparing to bat.}
  \label{fig:teaser}
\end{teaserfigure}
}

%
\maketitle

\section{Introduction}
Stream data processing has received much attention from the database community, due to its wide spectrum of real-world applications such as online clustering \cite{hyde2017fully}, intrusion detection \cite{dhanabal2015study}, sensor data monitoring \cite{abadi2004integration}, object identification \cite{hong2016explicit}, location-based services \cite{hong2012discovering}, IP network traffic analysis \cite{fusco2010net}, Web log mining \cite{carbone2015apache}, moving object search \cite{yu2013online}, event matching \cite{song2017matching}, and many others. In these applications, data objects from streams (e.g., sensory data samples) may sometimes contain missing attributes, for various reasons like packet losses, transmission delays/failures, and so on. It is therefore rather challenging to manage and process such streams with incomplete data effectively and efficiently.

In this paper, we will study the \textit{join} operator between incomplete data streams (i.e., streaming objects with missing attributes), which has real applications such as network intrusion detection, online clustering, sensor networks, and data integration. 


We have the following motivation example for the join over incomplete data streams in the application of network intrusion detection.


\setlength{\textfloatsep}{1pt}
\begin{figure}[t!]
\centering
\hspace{-1ex}\includegraphics[scale=0.22]{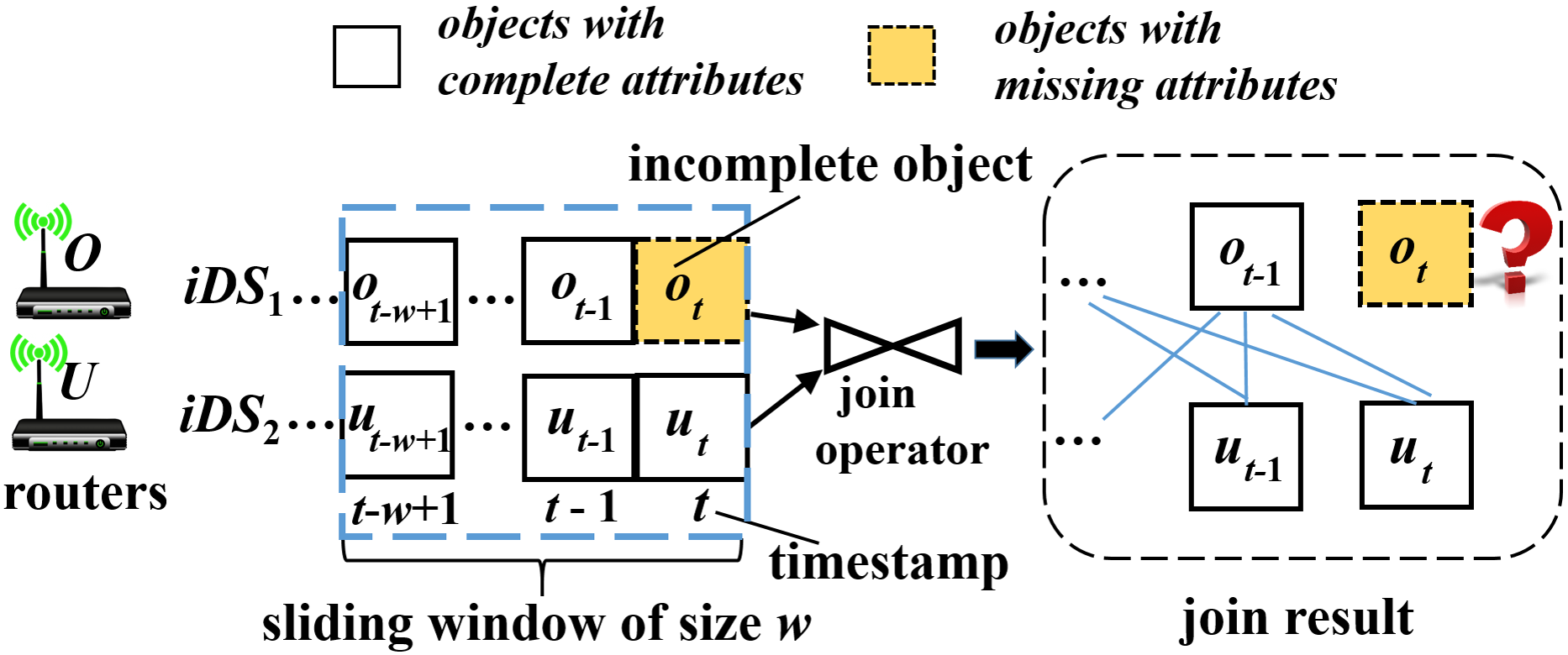}\vspace{-1ex}
\caption{\small The join operator over incomplete data streams for monitoring network intrusion.}
\label{fig:fig1}
\end{figure}



\begin{table}[t!]
\centering\scriptsize\vspace{-2ex}
\caption{\small Incomplete data streams, $iDS_1$ and $iDS_2$, in Figure \ref{fig:fig1}.}\label{table:table1}\vspace{-2ex}
\begin{tabular}{|c|c||c|c|c|}
\hline
\textbf{router} & \textbf{object ID from} &  [$A$] \textbf{No. of con-} & \textbf{[$B$] connection } & [$C$] \textbf{transferred}  \\
\textbf{ID} & \textbf{stream $iDS_1$ or $iDS_2$} & \textbf{nections ($\times 10^3)$} & \textbf{duration (min)} & \textbf{data size (GB)} \\
\hline
\hline
\multirow{4}{*}{$O$} & $o_{t-w+1}$ &  0.1 & 0.1 & 0.1 \\
 & $...$ & ...  & ... & ... \\
 & $o_{t-1}$ &  0.2 & 0.1 & 0.2 \\
 & $o_t$ &  0.4 & 0.3 & {\bf \color{blue}$-$}
 \\\hline\hline
\multirow{4}{*}{$U$} & $u_{t-w+1}$ &  0.2 & 0.1 & 0.1 \\
 & $...$ & ...  & ... & ... \\
 & $u_{t-1}$ &  0.2 & 0.2 & 0.2  \\
 & $u_t$ &  0.3 & 0.3 & 0.2 
 \\\hline
\end{tabular}
\end{table}

\nop{

\begin{table}[t!]
\centering\scriptsize
\caption{\small Euclidean distances between objects from streams $iDS_1$ and $iDS_2$ in Figure \ref{fig:fig1} and Table \ref{table:table1} at timestamps $t-1$ and $t$.}\vspace{-2ex}
\begin{tabular}{|c|c||c|}
\hline
\textbf{router ID} &  \textbf{router ID} & \textbf{Euclidean distance} \\
\hline
\hline
 $o_{t-1}$ &  $u_{t-1}$ &  0.224  \\\hline
$o_t$ &  $u_t$ & {\bf ?}   \\\hline
\end{tabular}
\label{table:table2}
\end{table}

}

\begin{example} {\bf (Monitoring Network Intrusion)} Figure \ref{fig:fig1} illustrates two critical routers, $O$ and $U$, in an IP network, from which we collect statistical (log) attributes in a streaming manner, for example, \textit{No. of connections}, \textit{the connection duration}, and \textit{the transferred data size}. In practice, due to packet losses, network congestion/delays, or hardware failure, we may not always obtain all attributes from each router. As an example in Table \ref{table:table1}, the \textit{transferred data size} of router $o_t$ is missing (denoted as ``-'') at timestamp $t$. As a result, stream data collected from each router may sometimes contain incomplete attributes.

One critical, yet challenging, problem in the network is to monitor network traffic, and detect potential network intrusion. 
If one router (e.g., $O$) is under the attack of network intrusion, we should quickly identify potential attacks in other routers, like $U$, at close timestamps, to which we may take actions for protecting the network security. In this case, it is very important to conduct the join over (incomplete) router data streams, and monitor similar patterns/behaviors from these two routers (e.g., $O$ and $U$). The resulting joining pairs can be used to effectively detect network intrusion events in routers.
\qquad $\blacksquare$

\label{example:example1}
\end{example}


In Example \ref{example:example1}, the join on incomplete router data streams monitors pairs of (potentially incomplete) objects from streams whose Euclidean distances are within some user-specified threshold. Due to the incompleteness of objects, it is rather challenging to accurately infer missing attributes, and effectively calculate the distance between 2 incomplete objects with missing attributes. For example, as depicted in Table \ref{table:table1}, it is not trivial how to compute the distance between object $o_t$ (with missing attribute $C$, \textit{transferred data size}) from router $O$ and any object $u_i$ (for $t-w+1\le i\le t$) from router $U$.



Inspired by the example above, in this paper, we formally define the \textit{join over incomplete data streams} (Join-iDS), which continuously monitors pairs of similar (incomplete) objects from two incomplete data streams with high confidences. In addition to the application of network intrusion detection (as shown in Example \ref{example:example1}), the Join-iDS problem is also useful for many other real applications, such as sensor data monitoring and data integration. 

One straightforward method to solve the Join-iDS problem is to conduct the imputation over data streams, followed by join processing over two imputed streams. However, this method is not that efficient, due to high imputation and joining costs, which may not suit for the requirements of stream processing (e.g., small response time). 

To tackle the Join-iDS problem efficiently and effectively, in this paper, we will propose an effective and adaptive imputation approach to turn incomplete data objects into complete ones, devise cost-model-based imputation indexes and a synopsis for data streams, and an efficient algorithm to simultaneously handle data imputation and Join-iDS processing.


\vspace{1ex}\noindent {\bf Differences from Prior Works.} While many prior works studied the \textit{join} operator over complete data streams \cite{das2003approximate,lin2015scalable} or uncertain data streams \cite{lian2010similarity,lian2009efficient}, they all assume that data streams are complete, and streaming objects do not have any missing attributes. To the best of our knowledge, no previous works considered the join operator over incomplete data streams (i.e., Join-iDS). To turn incomplete data records into complete ones, one straightforward way is to set the missing attribute values to 0, that is, ignoring the missing attribute values. However, this method may overestimate (underestimate) the distance between objects from data streams and cause wrong join results. Instead, in this paper, we will adopt \textit{differential dependency} (DD) rules \cite{song2011differential} to impute the possible values of missing attributes of data objects from incomplete data streams. 

Most importantly, in this paper, we will propose efficient Join-iDS processing algorithms to enable the data imputation and join processing at the same time, by designing cost-model-based and space-efficient index structures and efficient pruning strategies.

In this paper, we make the following major contributions:
\begin{enumerate}
\item We formalize a novel and important problem, \textit{join over incomplete data streams} (Join-iDS), in Section \ref{sec:problem_def}.

\item We propose effective and efficient cost-model-based data imputation techniques via DD rules in Section \ref{sec:imputation_of_io}.

\item We devise effective pruning strategies to reduce the Join-iDS search space in Section \ref{sec:pruning_strategies}.

\item We design an efficient Join-iDS processing algorithm via data synopsis/indexes in Section \ref{sec:join_over_iDS}. 

\item We evaluate through extensive experiments the performance of our Join-iDS approach on real/synthetic data in Section \ref{sec:exp_eval}.

\end{enumerate}

In addition, Section \ref{sec:related_work} reviews related works on the stream processing, differential dependency, join operator, and incomplete data management. Section \ref{sec:conclusions} concludes this paper. 

\section{Problem Definition}
\label{sec:problem_def}

In this section, we formally define the problem of the \textit{join over incomplete data streams} (Join-iDS), which takes into account the missing attributes in the process of the stream join.

\subsection{Incomplete data stream}
\label{subsec:iDS}

We first define two terms, \textit{incomplete data stream} and \textit{sliding window}, below. 


\begin{definition} \textbf{(Incomplete Data Stream)} 
An \textit{incomplete data stream}, $iDS$, contains an ordered sequence of objects, $(o_1, o_2, ...,$ $o_t,$ $...)$. Each object $o_i\in iDS$ arrives at timestamp $i$, and has $d$ attributes $A_j$ (for $1 \le j \le d$), some of which are missing, denoted as $o_i[A_j] =$ ``$-$''.
\label{def:iDS}
\end{definition}

In Definition \ref{def:iDS}, at each timestamp $i$, an object $o_i$ from incomplete data stream $iDS$ will arrive. Each object $o_i$ may be an incomplete object, containing some missing attributes $o_i[A_j]$.



Following the literature of data streams, in this paper, we consider the \textit{sliding window} model \cite{ananthakrishna2003efficient} over incomplete data stream $iDS$.

\begin{definition} \textbf{(Sliding Window, $W_t$)}
Given an incomplete data stream $iDS$, an integer $w$, and the current timestamp $t$, a \textit{sliding window}, $W_t$, contains an ordered set of the most recent $w$ objects from $iDS$, that is, $(o_{t-w+1}, o_{t-w+2}, ..., o_t)$.
\label{def:sw}
\end{definition}



In Definition \ref{def:sw}, the sliding window $W_t$ contains all objects from $iDS$ arriving within the time interval $[t-w+1, t]$. To incrementally maintain the sliding window, at a new timestamp $(t+1)$, a new sliding window $W_{t+1}$ can be obtained by adding the newly arriving object $o_{t+1}$ to $W_t$ and removing the old (expired) object $o_{t-w+1}$ from $W_t$.

Note that, the sliding window we adopt in this paper is the count-based one \cite{ananthakrishna2003efficient}. For other data models such as the time-based sliding window \cite{tao2006maintaining} (allowing more than one object arriving at each timestamp), we can easily extend our problem by replacing each object $o_i \in W_t$ with a set of objects arriving simultaneously at timestamp $i$, which we would like to leave as our future work.


\subsection{Imputation Over $iDS$}
\label{subsec:imputation_iDS}
In this paper, we adopt \textit{differential dependency} (DD) rules \cite{song2011differential} as our imputation approach for inferring the missing attributes of incomplete data objects from $iDS$. By using DD rules, incomplete data streams can be turned into \textit{imputed data streams}. We would like to leave the topics of considering other imputation methods (e.g., multiple imputation \cite{royston2004multiple}, editing rule \cite{fan2010towards}, relational dependency network \cite{mayfield2010eracer}, etc.) as our future work.



\vspace{1ex}\noindent {\bf Differential Dependency (DD).} The \textit{differential dependency} (DD) \cite{song2011differential} reveals correlation rules among attributes in data sets, which can be used for imputing the missing attributes in incomplete objects. As an example, given a table with 3 attributes $A$, $B$, and $C$, a DD rule can be in the form of $(A \to C, \{A.I,C.I\})$, where $A.I = [0,\epsilon_A]$ and $C.I = [0,\epsilon_C]$ are 2 distance constraints on attributes $A$ and $C$, respectively. Assuming $\epsilon_A = \epsilon_C = 0.1$, this DD rule implies that if any two objects $o_i$ and $o_j$ have their attribute $A$ within $\epsilon_A$-distance away from each other (i.e., $|o_i[A]-o_j[A]|\in [0, 0.1]$), then their values of attribute $C$ must also be within $\epsilon_C$-distance away from each other (i.e., $|o_i[C]-o_j[C]|\in [0, 0.1]$ holds).




Formally, we give the definition of the DD rule as follows.


\begin{definition} \textbf{(Differential Dependency, $DD$)} 
A \textit{differential dependency} (DD) rule is in the form of $(X \to A_j, \phi [X A_j])$, where $X$ is a set of \textit{determinant attributes}, $A_j$ is a \textit{dependent attribute} ($A_j\notin X$), and $\phi[Y]$ is a differential function to specify the distance constraints, $A_y.I$, on attributes $A_y$ in $Y$, where $A_y.I = [0,\epsilon_{A_y}]$.
\label{def:dd}
\end{definition}

In Definition \ref{def:dd}, given a DD rule $(X \to A_j, \phi [X A_j])$, if two data objects satisfy the differential function $\phi[X]$ on determinant attributes $X$, then they will have similar values on dependent attributes $A_j$. 

\vspace{1ex}\noindent {\bf Missing Values Imputation via DD and Data Repository $R$.} DD rules can achieve good imputation performance even in sparse data sets, since they tolerate differential differences between attribute values \cite{song2011differential}.

In this paper, we assume that a static data repository $R$ (containing complete objects without missing attributes) is available for imputing missing attributes from data streams. Given a DD $(X \to A_j, \phi[X\ A_j])$ and an incomplete object $o_i$, we can obtain possible values of missing attribute $o_i[A_j]$, by leveraging all complete objects $o_c \in R$ satisfying the differential function $\phi[X]$ w.r.t. attributes $X$ in $o_i$. 


Specifically, any imputed value $o_i[A_j] = val$ is associated with an existence probability $val.p$, defined as the fraction of complete objects $o_c$ with attribute $o_c[A_j] = val$ among all complete objects in $R$ satisfying the distance constraint $\phi[X]$ with $o_i$ on attribute(s) $X$.



\vspace{1ex}\noindent {\bf Imputed Data Stream.} With DDs and data repository $R$, we can turn incomplete data streams into imputed data streams, which is defined as follows.

\begin{definition} \textbf{(Imputed Data Stream, $pDS$)}
Given an incomplete data stream $iDS=(o_1,o_2,...,o_r,...)$, DD rules, and a static data repository $R$, the imputed (uncertain) data stream, $pDS=(o_1^p, o_2^p, ..., o_r^p, ...)$, is composed of imputed (probabilistic) objects, $o_i^p$, by imputing missing attribute values of incomplete objects $o_i$ via DDs and $R$. 

Each imputed object $o_i^p \in pDS$ contains a number of probabilistic instances, $o_{il}$, with existence confidences $o_{il}.p$, where instances $o_{il}$ are mutually exclusive and meet $\sum_{o_{il} \in o_i^p} o_{il}.p = 1$.
\label{def:pDS}
\end{definition}

In Definition \ref{def:pDS}, each object $o_i^p$ in $pDS$ is complete, containing a set of probabilistic instances $o_{il}$. In this paper, for each instance $o_{il}$, we calculate its existence probability $o_{il}.p$ as the product of confidences $val.p$ of $d$ attribute values $val$ of $o_{il}$.

\vspace{1ex} \noindent {\bf Possible Worlds Over $pDS$.} We consider \textit{possible worlds} \cite{dalvi2007efficient}, $pw(W_t)$, over the sliding window, $W_t$, of imputed data stream $pDS$, which are materialized instances of the sliding window that may appear in reality.

\begin{definition} \textbf{(Possible Worlds of the Imputed Data Stream, $pw(W_t)$)}
Given a sliding window $W_t$ of an imputed data stream $pDS$, a possible world, $pw(W_t)$, is composed of some object instances $o_{il}$, where these instances $o_{il}$ covers all imputed objects $o_i^p \in W_t$ and each instance comes from different imputed objects $o_i^p \in W_t$. 

The appearance probability, $Pr\{pw(W_t)\}$, of each possible world $pw(W_t)$ can be calculated by:\vspace{-2ex}

\begin{eqnarray}
Pr\{pw(W_t)\} = \prod_{o_{il} \in pw(W_t)} o_{il}.p.
\label{eq:eq1}\vspace{-3ex}
\end{eqnarray}

\label{def:pw_pDS}\vspace{-2ex}
\end{definition}

In Definition \ref{def:pw_pDS}, each imputed object $o_i^p \in W_t$ contributes to one potential instance $o_{il}$ to $pw(W_t)$, making each possible world $pw(W_t)$ a combination of instances from imputed objects in sliding window $W_t$.

\subsection{Join Over Incomplete Data Streams}
\label{subsec:Join_iDS}

\vspace{1ex}\noindent {\bf The Join-iDS Problem.} Now, we are ready to formally define the \textit{join over incomplete data streams} (Join-iDS).

\begin{definition} {\textbf{(Join Over Incomplete Data Streams, Join-iDS)}} Given two incomplete data streams, $iDS_1$ and $iDS_2$, a distance threshold $\epsilon$, a current timestamp $t$, and a probabilistic threshold $\alpha$, the \textit{join over incomplete data streams} (Join-iDS) continuously monitors pairs of incomplete objects $o_x$ and $o_y$ within sliding windows $W_{1t} \in iDS_1$ and $W_{2t} \in iDS_2$, respectively, such that they are similar with probabilities, $Pr_{Join\text{-}iDS}(o_x^p, o_y^p)$, greater than threshold $\alpha$, that is,\vspace{-2ex}

\begin{eqnarray}
&&Pr_{Join\text{-}iDS}(o_x^p, o_y^p)= Pr\{dist(o_x^p, o_y^p) \leq \epsilon\}\label{eq:eq2}\\
&=& \sum_{\forall pw(W_{1t})} \sum_{\forall pw(W_{2t})} Pr\{pw(W_{1t})\} \cdot Pr\{pw(W_{2t})\}\notag \\
&&  \cdot \chi\big(dist(o_{xl},o_{yg}) \leq \epsilon \text{ }|\text{ } o_{xl}\in pw(W_{1t}), o_{yg}\in pw(W_{2t})\big)
\ge \alpha, \notag\vspace{-3ex}
\end{eqnarray}
\noindent where $o_{xl}$ and $o_{yg}$ are instances of the imputed objects $o_x^p$ and $o_y^p$, respectively, $dist(\cdot, \cdot)$ is a Euclidean distance function, and function $\chi(z)$ returns 1, if $z = true$ (or 0, otherwise).
\label{def:Join-iDS}
\end{definition}



In Definition \ref{def:Join-iDS}, at timestamp $t$, Join-iDS will retrieve all pairs of incomplete objects, $(o_x, o_y)$, such that their distance is within $\epsilon$ threshold with Join-iDS probabilities, $Pr_{Join\text{-}iDS}(o_x^p, o_y^p)$ (as given by Eq.~(\ref{eq:eq2})), greater than or equal to $\alpha$, where $o_x^p \in W_{1t}$ and $o_y^p \in W_{2t}$. In particular, the Join-iDS probability, $Pr_{Join\text{-}iDS}(o_x^p, o_y^p)$, in Eq.~(\ref{eq:eq2}) is given by summing up probabilities that object instances $o_{xl}$ and $o_{yg}$ are within $\epsilon$-distance in possible worlds, $pw(W_{1t})$ and $pw(W_{2t})$.


\vspace{1ex}\noindent {\bf Challenges.} There are three major challenges to tackle the Join-iDS problem. First, existing works often assume that objects from data streams are either complete \cite{das2003approximate,lin2015scalable} or uncertain \cite{lian2010similarity,lian2009efficient}, and this assumption may not always hold in practice, due to reasons such as transmission delay or packet losses. Moreover, it is also non-trivial to obtain possible values of missing attributes. To our best knowledge, no prior work has studied the join operator over incomplete data streams. Thus, we should specifically design effective and efficient imputation strategies to infer incomplete objects from $iDS_1$ and $iDS_2$.


Second, it is very challenging to efficiently solve the Join-iDS problem under \textit{possible worlds} \cite{dalvi2007efficient} semantics. The direct computation of Eq.~(\ref{eq:eq2}) (i.e., materializing all possible worlds of two incomplete data streams) has an exponential time complexity, which is inefficient, or even infeasible. Thus, we need to devise efficient approaches to reduce the search space of our Join-iDS problem.

Third, it is not trivial how to efficiently and effectively process the join operator over data streams with incomplete objects, which includes data imputation and join processing over imputed data streams. To efficiently handle the Join-iDS problem, in this paper, we perform data imputation and join processing at the same time. Therefore, we need to propose efficient Join-iDS processing algorithms, supported by effective pruning strategies and indexing mechanism.


\begin{table}\hspace{-2ex}
{\small\scriptsize
    \caption{\small Symbols and descriptions.}
    \label{symbols_and_descriptions}
    \begin{tabular}{l|l} \hline
    {\bf Symbol} & \qquad\qquad\qquad\qquad{\bf Description} \\ \hline \hline
    $iDS$ ($iDS_1$ or $iDS_2$)  & an incomplete data stream \\ \hline
    $pDS$  & an imputed (probabilistic) data stream \\ \hline
    $W_{1t}$ (or $W_{2t}$)   & the most recent $w$ objects from stream $iDS_1$ (or $iDS_2$) at timestamp $t$ \\ \hline
    $w$ & the size of the sliding window \\ \hline
    $pw(W_t)$ & a possible world of imputed (probabilistic) objects in sliding window $W_t$ \\ \hline    
    $o_i$ ($o_x$ or $o_y$)   & an (incomplete) object from stream $iDS$ ($iDS_1$ or $iDS_2$)\\ \hline    
    $o_i^p$ & an imputed probabilistic object of $o_i$ in the imputed stream $pDS$ \\ \hline
    $R$ & a static (complete) data repository \\ \hline
    $Lat_j$ & an imputation lattice for DDs with dependent attribute $A_j$ \\ \hline
    $I_j$ & an index built over $R$ for imputing attribute $A_j$ \\ \hline
    $\epsilon$-grid & a data synopsis containing objects $o_x^p$ and $o_y^p$ from streams\\ \hline
    $JS$ & a join set containing object pairs $(o_x^p, o_y^p)$\\ \hline
    \end{tabular}\vspace{3ex}  
}
\end{table}

\subsection{Join-iDS Processing Framework}


Algorithm \ref{alg:Join_iDS_framework} illustrates a framework for Join-iDS processing, which consists of three phases. In the first \textit{pre-computation phase}, we offline establish \textit{imputation lattices} $Lat_j$ (for $1\le j\le d$), and build imputation indexes $I_j$ over a historical repository $R$ for imputing attribute $A_j$ (lines 1-2). Then, in the \textit{imputation and Join-iDS pruning} phase, we online maintain a data synopsis, called $\epsilon$-grid, over objects $o_x^p$ ($o_y^p$) from sliding window $W_{1t}$ ($W_{2t}$) of each incomplete data stream. In particular, for each expired object $o_x^{'}$ ($o_y^{'}$), we remove it from sliding window $W_{1t}$ ($W_{2t}$), update the $\epsilon$-grid, and update the join set, $JS$, w.r.t. object $o_x^p$ ($o_y^p$) (lines 3-6); for each newly arriving object $o_x$ ($o_y$), we will impute it by traversing indexes $I_j$ over $R$ with the help of DD rules (selected by the \textit{imputation lattice} $Lat_j$), prune false join objects $o_y^p \in W_{2t}$ ($o_x^p \in W_{1t}$) via the $\epsilon$-grid, and insert the imputed object $o_x^p$ ($o_y^p$) into the $\epsilon$-grid (lines 7-9). Finally, in the \textit{Join-iDS refinement phase}, we will calculate the join probabilities between $o_x^p$ ($o_y^p$) and each non-pruned object $o_y^p\in \epsilon$-grid ($o_x^p\in \epsilon$-grid), and return the join results $JS$ for all objects $o_v \in W_{1t}$ ($W_{2t}$) (lines 10-11).

Table~\ref{symbols_and_descriptions} depicts the commonly-used symbols and their descriptions in this paper.

\begin{algorithm}[t!]\scriptsize
\KwIn{two incomplete data streams $iDS_1$ and $iDS_2$, a static (complete) data repository $R$, current timestamp $t$, an timestamp interval $w$, a distance threshold $\epsilon$, and a probabilistic threshold $\alpha$}
\KwOut{a join result set, $JS$, over $W_{1t}$ and $W_{2t}$}
\tcp{Pre-computation Phase}
offline establish \textit{imputation lattice}, $Lat_j$, based on detected DDs from $R$

offline construct imputation indexes, $I_j$, over data repository $R$

\tcp{Imputation and Join-iDS Pruning Phase}
\For{each expired object $o_x^{'} \in iDS_1$ ($o_y^{'} \in iDS_2$) at timestamp $t$}{
    evict $o_x^{'}$ ($o_y^{'}$) from $W_{1t}$ ($W_{2t}$)
    
    update $\epsilon$-grid over $W_{1t}$ ($W_{2t}$)
    
    update join set $JS$ }

\For{each new object $o_x$ ($o_y$) arriving at $W_{1t}$ ($W_{2t}$)}{
    traverse index, $I_j$, over $R$ and $\epsilon$-grid, over $W_{1t}$ ($W_{2t}$) at the same time to simultaneously enable DD attribute imputation and join set preselection. 
    
    insert the data information of $o_x^p$ ($o_y^p$) into $\epsilon$-grid
    
}

\tcp{Join-iDS Refinement Phase}
calculate the join probabilities between $o_x^p$ ($o_y^p$) with each candidate $o_y^p \in \epsilon$-grid ($o_x^p\in \epsilon$-grid), and add the join pairs $(o_x^p,o_y^p)$ into $JS$

return the join sets, $JS$, for all objects $o_v \in W_{1t}$ ($W_{2t}$) as join results 

\caption{Join-iDS Processing Framework}
\label{alg:Join_iDS_framework}
\end{algorithm}

\vspace{1ex}
\section{Imputation of incomplete objects via DDs}
\label{sec:imputation_of_io}



\vspace{1ex}\noindent {\bf Data Imputation via DDs.} In Section \ref{subsec:imputation_iDS}, we discussed how to impute the missing attribute $A_j$ (for $1\le j\le d$) of an incomplete object $o_i$ by a single DD: $X\to A_j$. In practice, we may encounter multiple DDs with the same dependent attribute $A_j$, $X_1 \to A_j$, $X_2 \to A_j$, ..., and $X_l \to A_j$. In this case, one straightforward way is to combine all these DDs, that is, $(X_1X_2X_3...X_l \to A_j, \phi[X_1X_2...X_l A_j])$, to impute the missing attribute $o_i[A_j]$. By doing this, we may obtain a more selective query range, $X_1.I \land X_2.I\land ...\land X_l.I$, which may lead to not only more precise imputation results, but also the reduced imputation cost (i.e., with a smaller query range). However, to enable the imputation, such a combination has two requirements: (1) there should be at least one sample $o_c$ in data repository $R$ that satisfies the distance constraints $\phi[X_1X_2...X_l]$ w.r.t. $o_i$, and (2) incomplete object $o_i$ must have complete values on all attributes $X_1X_2...X_l$. Both requirements may not always hold, thus, alternatively we need to select a ``good'' subset of attributes $X_1X_2...X_l$ to impute $o_i[A_j]$.

\vspace{1ex}\noindent {\bf Imputation Lattice ($Lat_j$).} We propose a \textit{imputation lattice}, $Lat_j$ (for $1 \le j \le d$), which stores the combined DDs with all possible subsets of attributes $X_1X_2...X_l$, and can be used for selecting a ``good'' combined DD rule. In particular, each lattice $Lat_j$ has $l$ levels. Level 1 contains the $l$ original DD rules, with determinant attributes $X_1, X_2,$ $...$, and $X_l$; Level 2 has  $\big(^l_2\big)$ (i.e., $\frac{l\times (l-1)}{2}$) combined DDs, with determinant attributes such as $X_1X_2, X_1X_3,$ $...$, and $X_{l-1}X_l$; and so on. Finally, on Level $l$. there is only one combined DD rule, i.e., $X_1X_2X_3...X_l \to A_j$. 

\vspace{1ex}\noindent {\bf DD Selection Strategy.} Given an \textit{imputation lattice} $Lat_j$, we select a good DD rule from $Lat_j$ based on two principles. First, DDs on higher levels of $Lat_j$ (e.g., Level $l$) will have stronger imputation power than those on lower levels (e.g., Level $1$), since DDs on higher levels of $Lat_j$ tend to have more accurate imputation results and lower imputation cost. Second, for those DDs, $DD$, on the same level in $Lat_j$, we will offline estimate the expected numbers, $cnt(DD)$, of objects $o_c\in R$ that can be used for imputation via $DD$. We designed a cost model (via \textit{fractal dimension} \cite{belussi1998self}) for estimating $cnt(DD)$ in Appendix \ref{subsec:cost_DD_selection}. Since smaller $cnt(DD)$ indicates lower imputation cost and we need at least one sample for imputation, we rank DDs on the same level, first in increasing order for $cnt(DD) \ge 1$, and then in decreasing order for $cnt(DD) < 1$.

Given an incomplete object $o_i$ with missing attribute $A_j$, we traverse the lattice $Lat_j$ from Level $l$ to Level 1. On each level, we will access DDs in the offline pre-computed order as mentioned above. For each DD we encounter, we will online estimate the number of samples $o_c\in R$ for imputing attribute $A_j$ w.r.t. incomplete object $o_i$ (as given by Appendix \ref{subsec:cost_DD_selection}). If the expected number of objects for imputation is greater than or equal to 1, we will stop the lattice traversal, and use the corresponding DD for the imputation. 

Our proposed data imputation approaches via DDs are verified to be effective and efficient, whose empirical evaluation will be later illustrated in in Sections \ref{subsec:Join_effectiveness} and \ref{subsec:Join_efficiency}, respectively.

\section{Pruning Strategies}
\label{sec:pruning_strategies}

\subsection{Problem Reduction}
\label{subsec:problem_redu}


As given in Eq.~(\ref{eq:eq2}) of Section \ref{subsec:Join_iDS}, it is inefficient, or even infeasible, to compute join probabilities between two (incomplete) objects, $o_x$ and $o_y$, by enumerating an exponential number of \textit{possible worlds}. In this subsection, we reduce the problem of calculating the join probability, $Pr_{Join\text{-}iDS}(o_x^p,o_y^p)$, between $o_x$ and $o_y$ from possible-world level to that on object level, and rewrite Eq.~(\ref{eq:eq2}) as:\vspace{-2ex}

\begin{eqnarray}
&&Pr_{Join\text{-}iDS}(o_x^p,o_y^p) = Pr\{dist(o_x^p, o_y^p) \leq \epsilon\} \notag \\
&=& \sum_{\forall o_{xl} \in o_x^p} \sum_{\forall o_{yg} \in o_y^p} o_{xl}.p \cdot o_{yg}.p \cdot \chi(dist(o_{xl},o_{yg})\le \epsilon),
\label{eq:eq3}\vspace{-3ex}
\end{eqnarray}

\noindent where $o_{xl}.p$ and $o_{yg}.p$ are the existence confidences of instances $o_{xl} \in o_x^p$ and $o_{yg} \in o_y^p$, respectively, and function $\chi(\cdot)$ is given in Definition \ref{def:Join-iDS}.

In Eq.~(\ref{eq:eq3}), we consider all pairs, $(o_{xl},o_{yg})$, of instances $o_{xl} \in o_x^p$ and $o_{yg} \in o_y^p$, which is much more efficient than materializing all possible worlds, but still incurs $O(|o_x^p|\cdot |o_y^p|)$ cost, where $|o_i^p|$ is the number of instances in the imputed object $o_i^p$. Thus, in the sequel, we will design effective pruning rules to accelerate Join-iDS processing.


\begin{figure}
\centering 
\subfigure[][{\small object-level pruning}]{                    
\scalebox{0.18}[0.18]{\includegraphics{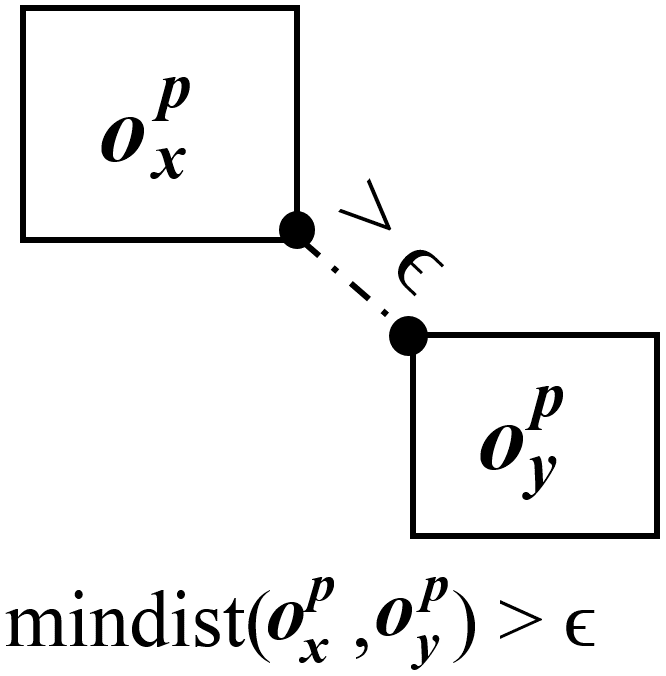}}\label{subfig:pruning1}          
}\qquad
\nop{
\subfigure[][{\small pivot pruning}]{
\scalebox{0.17}[0.17]{\includegraphics{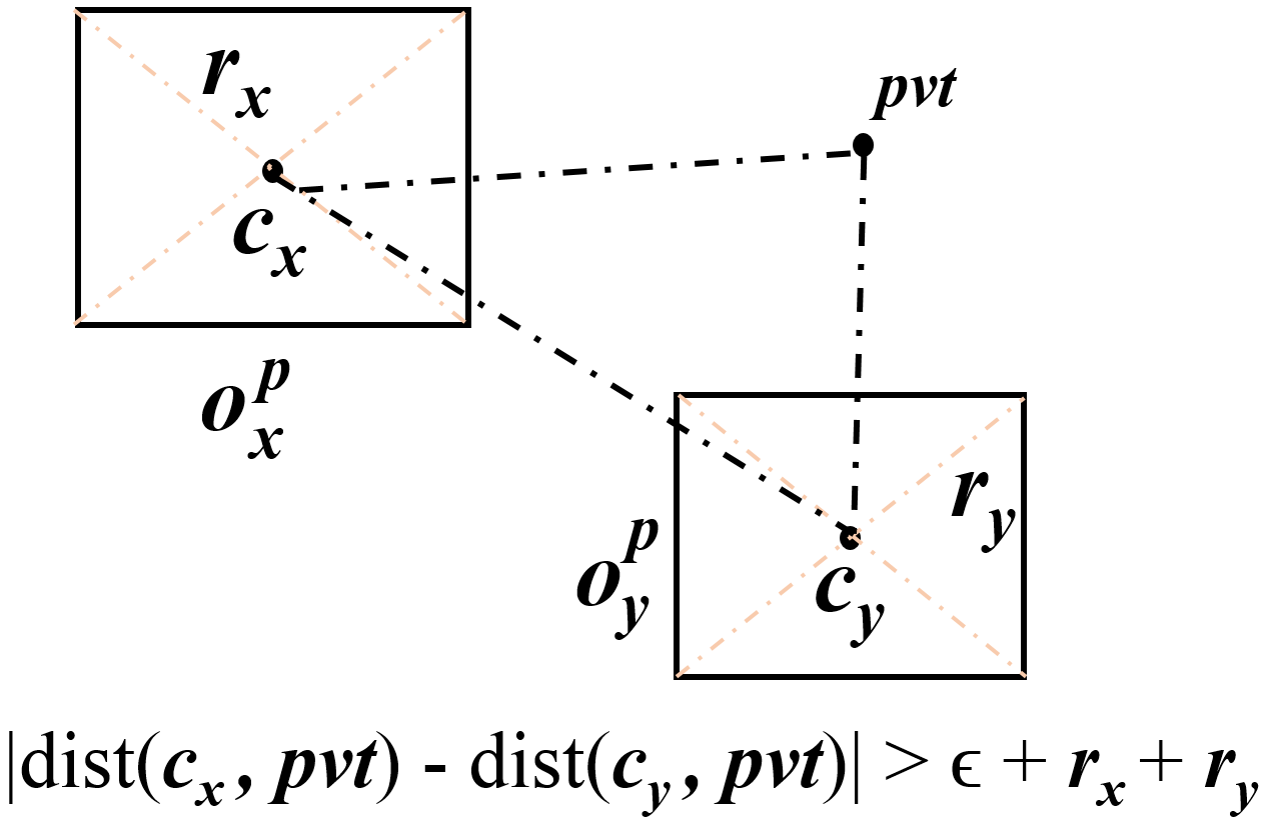}}\label{subfig:pruning2}       
}\qquad\qquad
}
\subfigure[][{\small sample-level pruning}]{
\scalebox{0.15}[0.15]{\includegraphics{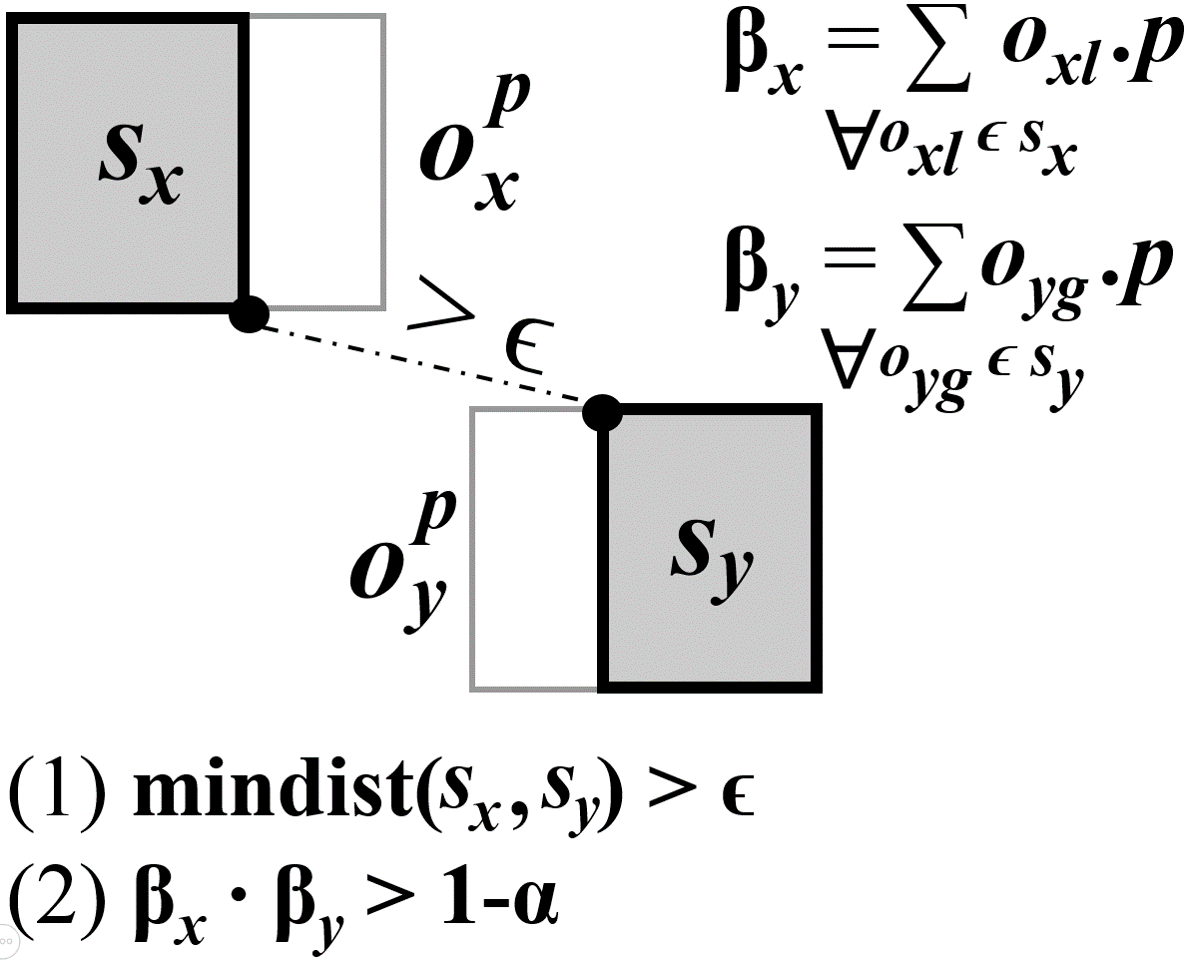}}\label{subfig:pruning3} 
}\vspace{-1ex}  
\caption{\small Illustration of pruning strategies.} \label{fig:pruning} \vspace{2ex}                           
\end{figure}

\subsection{Pruning Rules}
\label{subsec:pruning_rules}
Below, we propose two pruning strategies, \textit{object-level} and \textit{sample-level pruning}, to reduce the Join-iDS search space. The latter one will be used, if an object pair cannot be pruned by the former one.


\noindent {\bf Object-Level Pruning.} Given two incomplete objects $o_x \in W_{1t}$ and $o_y \in W_{2t}$, our first pruning rule, namely \textit{object-level pruning}, is to utilize the boundaries of the imputed objects $o_x^p$ and $o_y^p$, and filter out the object pair $(o_x, o_y)$ if their minimum possible distance $mindist(o_x^p,o_y^p)$ is greater than the distance threshold $\epsilon$. Here, the boundary of an imputed object $o_x^p$ (or called \textit{minimum bounding rectangle} (MBR)) encloses all instances of $o_x^p$ and has an imputed interval, $[o_x^p[A_j].min,o_x^p[A_j].max]$, for any missing attribute $A_j$ (for $1\le j\le d$).



\begin{lemma} {\bf (Object-Level Pruning)}
Given two incomplete objects, $o_x$ and $o_y$, from sliding windows $W_{1t}\in iDS_1$ and $W_{2t}\in iDS_2$, respectively, if $mindist(o_x^p,o_y^p) > \epsilon$ holds, then object pair $(o_x, o_y)$ can be safely pruned.
\label{lemma:lem1}
\end{lemma}

\nop{
\noindent where the minimum distance $dist(o_x^p,o_y^p).min$ can be calculated in Eq. (\ref{eq:eq4}).
\begin{eqnarray}
dist(o_x^p,o_y^p).min = \sqrt{\sum_{j=1}^d \omega^2(o_x^p,o_y^p,A_j)}
\label{eq:eq4}
\end{eqnarray}
\noindent where $\omega(o_x^p,o_y^p,A_j)$ is the function to obtain the minimum distance between imputed objects $o_x^p$ and $o_y^p$ on attributes $A_j$, which is calculated via Eq. (\ref{eq:eq5}).
\begin{eqnarray}
\begin{cases}
o_x^{p}[A_j].min-o_y^p[A_j].max& if\ o_y^p[A_j].max < o_x^p[A_j].min\\
o_y^p[A_j].min-o_x^p[A_j].max& if\ o_x^p[A_j].max < o_y^p[A_j].min\\
0& \text{otherwise}
\end{cases}
\label{eq:eq5}
\end{eqnarray}
}
\begin{proof}
Please refer to Appendix \ref{subsec:proof_lem1}.
\end{proof} 


Figure \ref{subfig:pruning1} illustrates an example of Lemma \ref{lemma:lem1}. Intuitively, if $mindist(o_x^p,o_y^p)>\epsilon$ holds, then two imputed objects (MBRs), $o_x^p$ and $o_y^p$, are far away from each other, and any instance pair from them cannot be joined (i.e., $Pr_{Join\text{-}iDS}(o_x^p,o_y^p)=0$). Thus, object pair $(o_x, o_y)$ can be safely pruned.

\vspace{1ex}\noindent {\bf Sample-Level Pruning.} The object-level pruning rule cannot filter out object pairs with non-zero Join-iDS probabilities $Pr_{Join\text{-}iDS}(o_x,$ $o_y)$ ($\in (0, \alpha)$). Thus, we present a \textit{sample-level pruning} method, which aims to rule out those false alarms with low Join-iDS probabilities, by considering instances of imputed objects $o_x^p$ and $o_y^p$.






\begin{lemma} {\bf (Sample-Level Pruning)}
Given two incomplete objects $o_x \in W_{1t}$ and $o_y \in W_{2t}$, and two sub-MBRs, $s_x \subseteq o_x^p$ and $s_y \subseteq o_y^p$, the object pair, $(o_x, o_y)$, can be safely pruned, if $mindist(s_x,s_y) > \epsilon$ and $\beta_x \cdot \beta_y > 1- \alpha$ hold, \noindent where $\beta_x = \sum_{\forall o_{xl} \in s_x}$ $o_{xl}.p$ is the summed probability that instances $o_{xl} \in o_x^p$ fall into sub-MBR $s_x$ (the same for $\beta_y$ w.r.t. $s_y$).
\label{lemma:lem3}
\end{lemma}

\begin{proof}
Please refer to Appendix \ref{subsec:proof_lem3}.
\end{proof} 

Figure \ref{subfig:pruning3} shows an example the \textit{sample-level pruning} in Lemma \ref{lemma:lem3}, which considers instances of imputed objects $o_x^p$ and $o_y^p$, and uses their sub-MBRs, $s_x \subseteq o_x^p$ and $s_y \subseteq o_y^p$, to enable the pruning, where $s_x$ (or $s_y$) is a sub-MBR such that object $o_x$ (or $o_y$) falls into $s_x$ (or $s_y$) with probability $\beta_x$ (or $\beta_y$). Intuitively, if $mindist(s_x,s_y) > \epsilon$ and $\beta_x \cdot \beta_y > 1-\alpha$ hold, then we can prove that the object pair $(o_x, o_y)$ has low join probability (i.e., $<\alpha$), and can be safely pruned.

\section{Join over incomplete data streams}
\label{sec:join_over_iDS}

In this section, we first design a data synopsis for incomplete data streams and imputation indexes over data repository $R$, and then propose an efficient Join-iDS processing algorithm to retrieve the join results via synopsis/indexes.

\begin{figure}[t!]
\centering
\hspace{2ex}\includegraphics[scale=0.16]{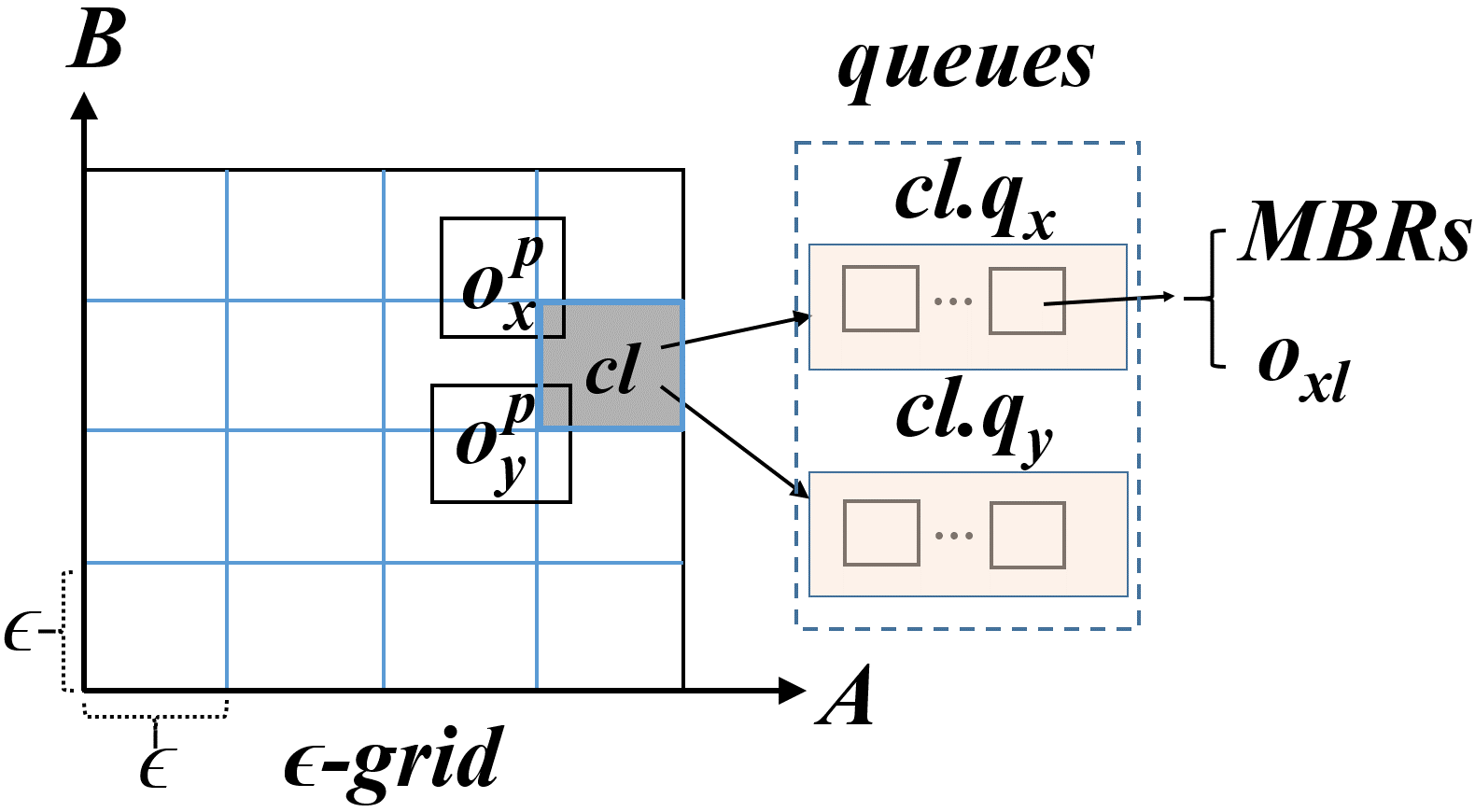}\vspace{-1ex}
\caption{\small Illustration of a 2D $\epsilon$-grid over incomplete data streams.}\label{fig:e_grid}
\end{figure}

\subsection{Grid Synopsis and Imputation Indexes}
\label{subsec:grid_synopsis}


\noindent {\bf $\epsilon$-Grid Over Imputed Data Streams.} We will incrementally maintain a data synopsis, namely \textit{$\epsilon$-grid}, over (imputed) objects $o_x^p$ and $o_y^p$ from sliding windows $W_{1t}\in iDS_1$ and $W_{2t}\in iDS_2$, respectively. Specifically, to construct the \textit{$\epsilon$-grid}, we divide the data space into equal grid cells with side length $\epsilon$ along each dimension (attribute $A_j$). Each cell, $cl$, is associated with two queues, $cl.q_x$ and $cl.q_y$, which sequentially store imputed objects $o_x^p\in W_{1t}$ and $o_y^p\in W_{2t}$, respectively, that intersect with this cell $cl$. Each imputed object $o_x^p$ (or $o_y^p$) contains information as follows:
\begin{enumerate}
\item a set of currently accessed MBR nodes, $MBRs$, in the R$^*$-tree over data repository $R$ for imputation (as will be discussed later in this subsection), or;
\item a set of instances, $o_{xl}$ (or $o_{yg}$), in $o_x^p$ (or $o_y^p$).
\end{enumerate}




Figure \ref{fig:e_grid} illustrates an example of the $\epsilon$-grid with two attributes $A$ and $B$. The $\epsilon$-grid divides the 2D data space into 16 ($=4\times 4$) cells, each with side length $\epsilon$. If imputed object $o_x^p$ (or $o_y^p$) intersects with a cell $cl$, then this object will be stored in a queue $cl.q_x$ (or $cl.q_y$) pointed by cell $cl$.


\begin{figure}[t!]
\centering
\hspace{2ex}\includegraphics[scale=0.22]{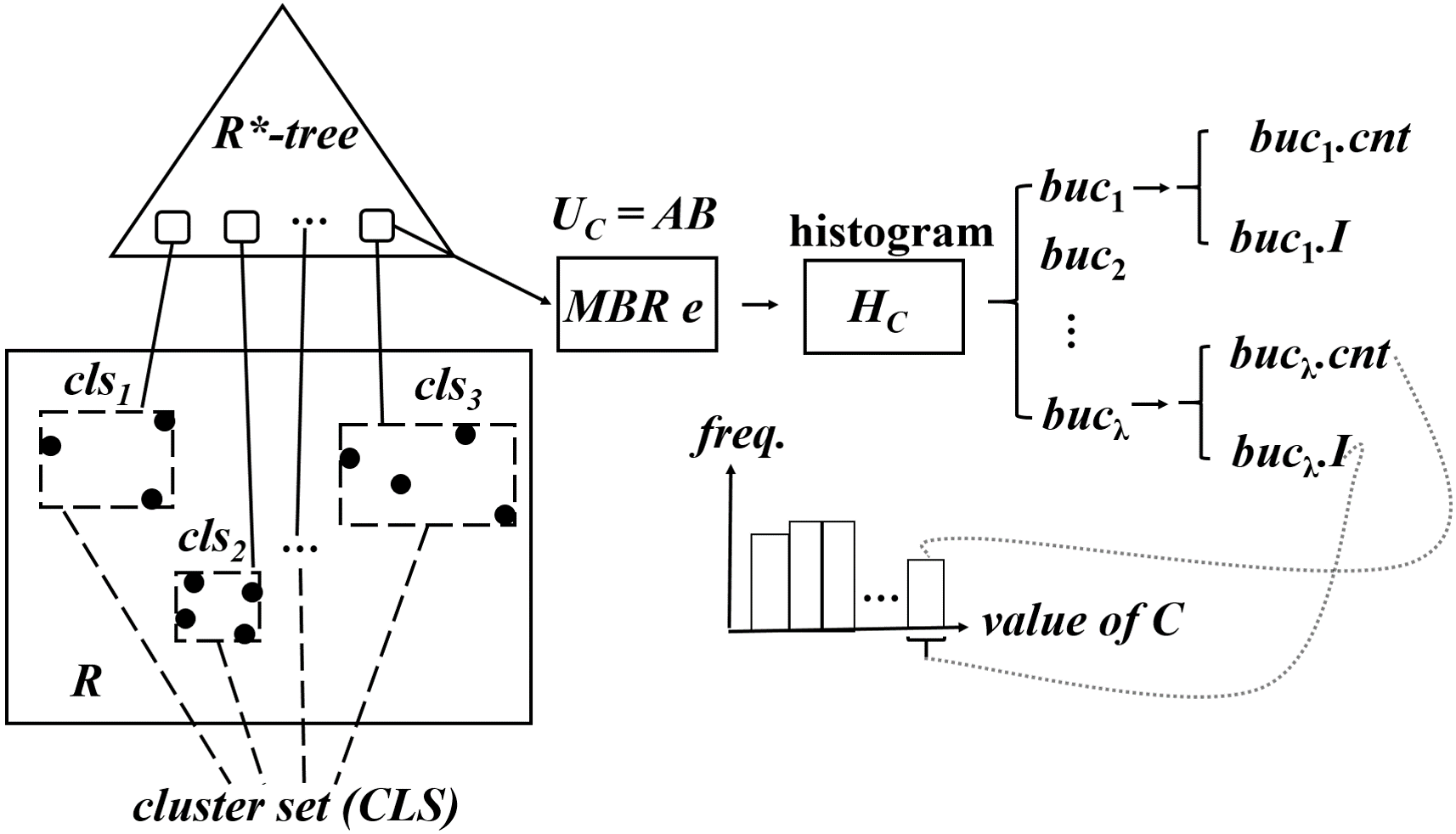}\vspace{-1ex}
\caption{\small Imputation Index over repository $R$ given a DD $AB \to C$ ($U_C = AB$).}
\label{fig:index_structure}\vspace{2ex}
\end{figure}

\vspace{1ex}\noindent {\bf Imputation Indexes Over Data Repository $R$.} To enable fast imputation, we devise $d$ indexes, $I_j$ (for $1\le j\le d$), each of which will have the best imputation power for a possibly missing attribute $A_j$. Specifically, assume that the combined DD rule on Level $l$ of the \textit{imputation lattice} $Lat_j$ (see Section \ref{sec:imputation_of_io}) is $X_1X_2...X_l \to A_j$. Then, we let $U_j = X_1 \cup X_2 \cup ... \cup X_l$, and construct a variant of R$^*$-tree \cite{beckmann1990r} over attributes $U_j$ of data repository $R$. 

We divide complete objects $o_c$ in data repository $R$ into $n$ clusters, $cls_1 \sim cls_n$, and insert them into the R$^*$-tree, where each cluster has the size within $[m, M]$. We design a specific cost model to select a good cluster set. Please refer to Appendix \ref{subsec:cost_for_cluster} for details. 

Moreover, each node $e$ in R$^*$-tree stores a histogram, $H_{A_j}$, over dependent attribute $A_j$, which stores a summary of complete objects $o_c$ in $e$, where $H_{A_j}$ is divided into $\lambda$ buckets, $buc_f$ ($1\le f\le \lambda$), with consecutive bucket intervals $buc_f.I=[buc_f.A_j^{-},buc_f.A_j^{+}]$ (i.e., $buc_f.A_j^{+} = buc_{f+1}.A_j^{-}$), and each bucket $buc_f$ contains all ($buc_f.cnt$) objects $o_c \in e$ with attribute values $o_c[A_j]$ within the interval $buc_f.I$.

Figure \ref{fig:index_structure} gives an example of a table with 3 attributes $A, B$ and $C$, and two DD rules, $A\to C$ and $B\to C$, with dependent attribute $C$. We construct an index $I_C$ for imputing attribute $C$, where $U_j = AB$ and $A_j=C$. In this example, we first put complete objects $o_c\in R$ into some clusters (e.g., $cls_1$), and then insert these clusters into an R$^*$-tree as leaf nodes. As shown in Figure \ref{fig:index_structure}, each node $e$ is divided into $\lambda$ buckets, $buc_1 \sim buc_{\lambda}$, based on the distribution of dependent attribute $C$, where each bucket $buc_f$ ($1\le f\le \lambda$) contains the count $buc_f.cnt$ of objects $o_c$ and the interval $buc_f.I$ of values $o_c[C]$ on attribute $C$ of $o_c$ in the bucket $buc_f$.

\subsection{Join-iDS Processing via $\epsilon$-Grid}
\label{subsec:Join_via_grid}




\nop{

{\color{Xiang} Algorithm 2, line 4, it is not a sentence. check my revision.

line 5, when do we impute to instance level, and when to MBR level?

line 7, it is not clear what does it mean by "(cannot be pruned by Lemma \ref{lemma:lem1})"?

\color{Weilong} (professor, it is revised.)
}

}

\begin{algorithm}[t!]\scriptsize
\KwIn{a join set $JS$, a $\epsilon$-grid synopsis, imputation indexes $I_j$ over $R$, and new objects $o_x$ and $o_y$ from $W_{1t}\in iDS_1$ and $W_{2t}\in iDS_2$}
\KwOut{an dynamically updated $JS$ and $\epsilon$-grid}
remove from $\epsilon$-grid those expired objects from streams $iDS_1$ and $iDS_2$

remove from $JS$ object pairs containing the expired objects 

obtain initial $o_x^p.MBR$ via R*-tree nodes in index $I_j$ and DD rules returned by $Lat_j$


\If{there exists some grid cell $cl \in\epsilon$-grid with nonempty queues $cl.q_y$, such that $mindist(cl, o_x^p.MBR)\leq \epsilon$ (via Lemma \ref{lemma:lem1})}{
    obtain instances $o_{xl}$ of imputed object $o_x^p$ by accessing objects $o_c$ via indexes $I_j$
    
    update $o_x^p.MBR$
}

\For{each cell $cl \in\epsilon$-grid with non-empty queue $cl.q_y$ that cannot be pruned via Lemma \ref{lemma:lem1}}{
     \If{$mindist(s_x,cl) \le \epsilon$ (via Lemma \ref{lemma:lem3})}{
            \For{each unchecked object $o_y^p$ in queue $cl.q_y$ satisfying $mindist(o_x^p.MBR, o_y^p.MBR) \leq \epsilon$}{
         \If{$o_y^p$ is not completely imputed via indexes $I_j$}{
        impute $o_y^p$ to instance level via indexes $I_j$, and update $o_y^p.MBR$
        
         update those cells $cl\in \epsilon$-grid intersecting with $o_y^p.MBR$ 
        
    }
    
    \If{$mindist(s_x,s_y) \le \epsilon$ (via Lemma \ref{lemma:lem3})}{
        \If{$Pr_{Join\text{-}iDS}(o_x^p,o_y^p) \ge \alpha$ via Eq. (\ref{eq:eq3})}{
        add $(o_x^p,o_y^p)$ to $JS$
    }
    }

    }
        }
}

\For{each cell $cl$ intersecting with the MBRs of $o_x^p$}{
    add $o_x^p$ to queue $cl.q_x$
}

execute lines 3-17 symmetrically for new object $o_y \in iDS_2$

\caption{Join-iDS via $\epsilon$-grid}
\label{alg:join_via_grid}
\end{algorithm}


\noindent {\bf Join-iDS via $\epsilon$-Grid.} Denote $JS$ as a join set that records all join results, $(o_x^p, o_y^p)$, between two incomplete data streams. Algorithm \ref{alg:join_via_grid} performs the object imputation and join at the same time, and dynamically maintain the join set $JS$ (and $\epsilon$-grid as well).

\underline{Deletion of the expired objects.} At a new timestamp $t$, Algorithm \ref{alg:join_via_grid} will remove the expired objects from $\epsilon$-grid and those object pairs containing the expired objects from $JS$ (lines 1-2).



\underline{Object imputation and object-level pruning.}  Given a newly arriving incomplete object $o_x \in W_{1t}$, Algorithm \ref{alg:join_via_grid} will retrieve a query range $Q$ via a DD rule returned by the \textit{imputation lattice} $Lat_j$ (Section \ref{sec:imputation_of_io}), and obtain an initial MBR $o_x^p.MBR$, by accessing R$^*$-tree nodes that intersect with the query range $Q$ via imputation index $I_j$  (line 3). Then, we will check whether there are some cells, $cl$, in the $\epsilon$-grid, that may match with $o_x^p$ (via Lemma \ref{lemma:lem1}). In particular, if $mindist(cl,o_x^p.MBR)\leq \epsilon$ holds, we will further obtain instances $o_{xl}$ of imputed object $o_x^p$, and update (shrink) the MBR of $o_x^p.MBR$ (lines 4-6). 

\underline{Object imputation and sample-level pruning.} Next, for cells that cannot be pruned by Lemma \ref{lemma:lem1}, Algorithm \ref{alg:join_via_grid} will further check the minimum distance, $mindist(s_x,cl)$, between sub-MBRs $s_x \in o_x^p.MBR$ and cell $cl$ via the sample-level pruning (Lemma \ref{lemma:lem3}; lines 7-15). If $mindist(s_x,cl)\le \epsilon$ and queues $cl.q_y$ are non-empty, then we will check the minimum distance, $mindist(o_x^p.MBR, o_y^p.MBR)$, between imputed objects $o_x^p$ and each unchecked object $o_y^p$ in the queues $cl.q_y$ of cell $cl$ (lines 9-15). Note that, each object $o_y^p\in cl.q_y$ may have two possible imputation states: (1) object $o_y^p$ is represented by MBRs $o_y^p.MBR$, or (2) object $o_y^p$ is represented by some samples (the missing attributes are imputed from $R$). We call the first state ``not completely imputed'', while the second one ``completely imputed''. If $o_y^p$ is not completely imputed, we  will impute $o_y^p$ completely via indexes $I_j$, and update the cells in $\epsilon$-grid intersecting with $o_y^p.MBR$ (lines 10-12). Given both completely imputed objects $o_x^p$ and $o_y^p$, if $mindist(s_x, s_y)\le \epsilon$, we will use the sample-level pruning to prune the object pair $(o_x^p,o_y^p)$ (please refer to Appendix \ref{sec:Ij_for_sample_pruning} for the selection of sub-MBRs $s_x$ and $s_y$; line 13). If the object pair $(o_x^p,o_y^p)$ still cannot be pruned, then we will check the join probabilities, $Pr_{Join\text{-}iDS}(o_x^p,o_y^p)$, via Eq.~(\ref{eq:eq3}), and add actual join pairs to $JS$ (lines 14-15). 


\underline{Update of $\epsilon$-grid with new object $o_x^p$.} Algorithm \ref{alg:join_via_grid} then inserts new object $o_x^p$ into all queues, $cl.q_x$, in cells $cl$ (intersecting with the MBR of $o_x^p$) of the $\epsilon$-grid (lines 16-17).

Finally, similar to object $o_x$, we execute lines 3-17 for a newly arriving object $o_y$ from sliding window $W_{2t}\in iDS_2$ (line 18).

\vspace{1ex}\noindent {\bf Complexity Analysis.} The Join-iDS algorithm in Algorithm \ref{alg:join_via_grid} requires $O(|o_x^p| \cdot |o_y^p| \cdot |cl| \cdot |cl.q_y|)$ time complexity, where $|o_x^p|$ and $|o_y^p|$ are the average numbers of instances in imputed objects $o_x^p$ and $o_y^p$, respectively, $|cl|$ is the number of cells intersecting with the MBR $o_x^p.MBR$ of $o_x^p$, and $|cl.q_y|$ is the average number of objects $o_y^p$ within queues $cl.q_y$ of cells $cl$.

\vspace{1ex}\noindent {\bf Discussions on the Extension of Join-iDS to $n$ ($>2$) Incomplete Data Streams.} We can extend our Join-iDS problem over 2 incomplete data streams to multiple (e.g., $n>2$) incomplete data streams $iDS_1 \sim iDS_n$. We only need to update the $\epsilon$-grid, that is, increase the number of queues in each cell of $\epsilon$-grid from 2 to $n$. Within a cell in $\epsilon$-grid, each queue, $cl.q_i$ ($1\le i\le n$), stores objects from its corresponding incomplete data stream $iDS_i$. With the modified $\epsilon$-grid, at timestamp $t$, when a new object $o_x$ arrives, the imputed object $o_x^p$ will push its join pairs into $JS$, by accessing those objects from $(n-1)$ queues in each cell.

\section{Experimental Evaluation}
\label{sec:exp_eval}

\begin{table}[t!]
\centering\scriptsize
\caption{\small The parameter settings.} \label{table:exp_parameter_setting}\vspace{-2ex}
\begin{tabular}{|l|c|}
\hline
\qquad\qquad\qquad\qquad\textbf{Parameters} & \textbf{Values}\\
\hline
\hline
probabilistic threshold $\alpha$ &  0.1, 0.2, \textbf{0.5}, 0.8, 0.9 \\\hline
dimensionality $d$ & 2, 3, \textbf{4}, 5, 6, 10 \\\hline
distance threshold $\epsilon$ & 0.1, 0.2, \textbf{0.3}, 0.4, 0.5  \\\hline
the number, $|W_t|$, of valid objects in $iDS$ & 500, 1K, \textbf{2K}, 4K, 5K, 10K \\\hline
the size, $|R|$, of data repository $R$ & 10K, 20K, \textbf{30K}, 40K, 50K \\\hline
the number, $m$, of missing attributes & \textbf{1}, 2, 3 \\\hline
\end{tabular}
\end{table}

\subsection{Experimental Settings}
\label{subsec:exp_settings}

\noindent {\bf Real/Synthetic Data Sets.} We evaluate the performance of our Join-iDS approach on 4 real and 3 synthetic data sets. 

\underline{Real data sets.} We use Intel lab data\footnote{\scriptsize\url{http://db.csail.mit.edu/labdata/labdata.html}}, UCI gas sensor data for home activity monitoring\footnote{\scriptsize\url{http://archive.ics.uci.edu/ml/datasets/gas+sensors+for+home+activity+monitoring}}, US \& Canadian city weather data\footnote{\scriptsize\url{https://www.kaggle.com/selfishgene/historical-hourly-weather-data}}, and S\&P 500 stock data\footnote{\scriptsize\url{https://www.kaggle.com/camnugent/sandp500}}, denoted as $Intel$, $Gas$, $Weather$ and $Stock$, respectively. $Intel$ contains 2.3 million data, collected from 54 sensors deployed in Intel Berkeley Research lab on Feb. 28-Apr. 5, 2014; $Gas$ includes 919,438 samples from 8 MOX gas sensors, and humidity and temperature sensors; $Weather$ contains $45.3K$ historical weather (temperature) data for 30 US and Canadian Cities during 2012-2017; $Stock$ has $619K$ historical stock data for all companies found on the S\&P 500 index till Feb 2018. We extract 4 attributes from each of these 4 real data sets: temperature, humidity, light, and voltage from $Intel$; temperature, humidity, and resistance of sensors 7 and 8 from $Gas$; Vancouver, Portland, San Francisco, and Seattle from $Weather$; and open, high, low, close from $Stock$. We normalize the intervals of 4 attributes of each data sets into $[0,1]$. Then, as depicted in Table \ref{table:real_DDs}, we detected DD rules for each data set, by considering all combinations of determinant/dependent attributes over samples in data repository $R$ \cite{song2011differential}. 

\begin{table}[t!]
\centering\scriptsize\vspace{-3ex}
\caption{\small The tested data sets and their DD rules.} 
\label{table:real_DDs}\vspace{-2ex}
\begin{tabular}{|c|l|}
\hline
\textbf{Data Sets} & \qquad\qquad\qquad\qquad\quad \textbf{DD Rules}\\
\hline
\hline
$Intel$ & $humidity, light, voltage \to temperature$,\\
& \{$[0,0.0036], [0, 0.645], [0, 0.116], [0, 0.003]$\}\\\hline
$Gas$ &  $resistance8, temperature \to resistance7$,\\
& $\{[0, 0.0483], [0, 0.751], [0, 0.4]\}$ \\\hline
$Weather$ &  $Vancouver, Portland, San\ Francisco \to Seattle$,\\
& $\{[0, 0.003], [0, 0.382], [0, 0.539], [0, 0.014]\}$ \\\hline
$Stock$ &  $open, high, low \to close$,\\
& $\{[0, 0.037], [0, 0.013], [0, 0.014], [0, 0.01]\}$ \\\hline
$Uniform$& $A B C \to D$, $\{[0,0.01], [0,0.01], [0,0.01], [0,0.01]\}$\\
$Correlated$ & $A B \to E$, $\{[0,0.02], [0,0.02], [0,0.05]\}$ \\
$Anti$-$correlated$ & $A C \to F$, $\{[0,0.03], [0,0.03], [0,0.1]\}$\\
 & $B \to F$, $\{[0,0.02], [0,0.05]\}$
\\\hline
\end{tabular}
\end{table}

\underline{Synthetic data sets.} We produce 3 types of $d$-dimensional synthetic data \cite{borzsony2001skyline}, following uniform, correlated, and anti-correlated distributions, denoted as $Uniform$, $Correlation$, and $Anti$-$Correlation$, respectively. For each data distribution, we first generate 5,000 seeds, and then obtain the remaining objects based on DD rules in Table \ref{table:real_DDs}.

\underline{Incomplete data generation.} For real/synthetic data above, we randomly select $m$ dependent attributes (e.g., $temperature$ of $Intel$) for objects from incomplete data streams $iDS_1$ and $iDS_2$, and set them as missing (=``-''). Note that, for each real/synthetic data set, we divide it into three subsets, which are corresponding to incomplete data streams $iDS_1$ and $iDS_2$, and complete data repository $R$, respectively.


\noindent {\bf Competitor.} We compare our Join-iDS approach with two baseline approaches, namely $DD+ASP$ and $DD+\epsilon$\text{-}$grid$, which which first impute incomplete objects via DDs and data repository $R$, and then obtain join results by considering pairwise objects from imputed data streams via the join algorithm in \cite{lian2010similarity} and our proposed join approach via $\epsilon$-grid synopsis, respectively.


\noindent {\bf Measure.} We report the \textit{wall clock time}, which is the total CPU time to perform the data imputation (via DDs and imputation indexes) and join processing (via $\epsilon$-grid) at the same time. 


\noindent {\bf Parameter Settings.} Table \ref{table:exp_parameter_setting} depicts experimental settings, where default parameter values are in bold. Each time we test one parameter, while setting other parameters to their default values. We run our experiments on a PC with Intel(R) Core(TM) i7-6600U CPU 2.70 GHz and 32 GB memory. All algorithms were implemented by C++.

\begin{figure}[t!]
\centering
\scalebox{0.21}[0.21]{\includegraphics{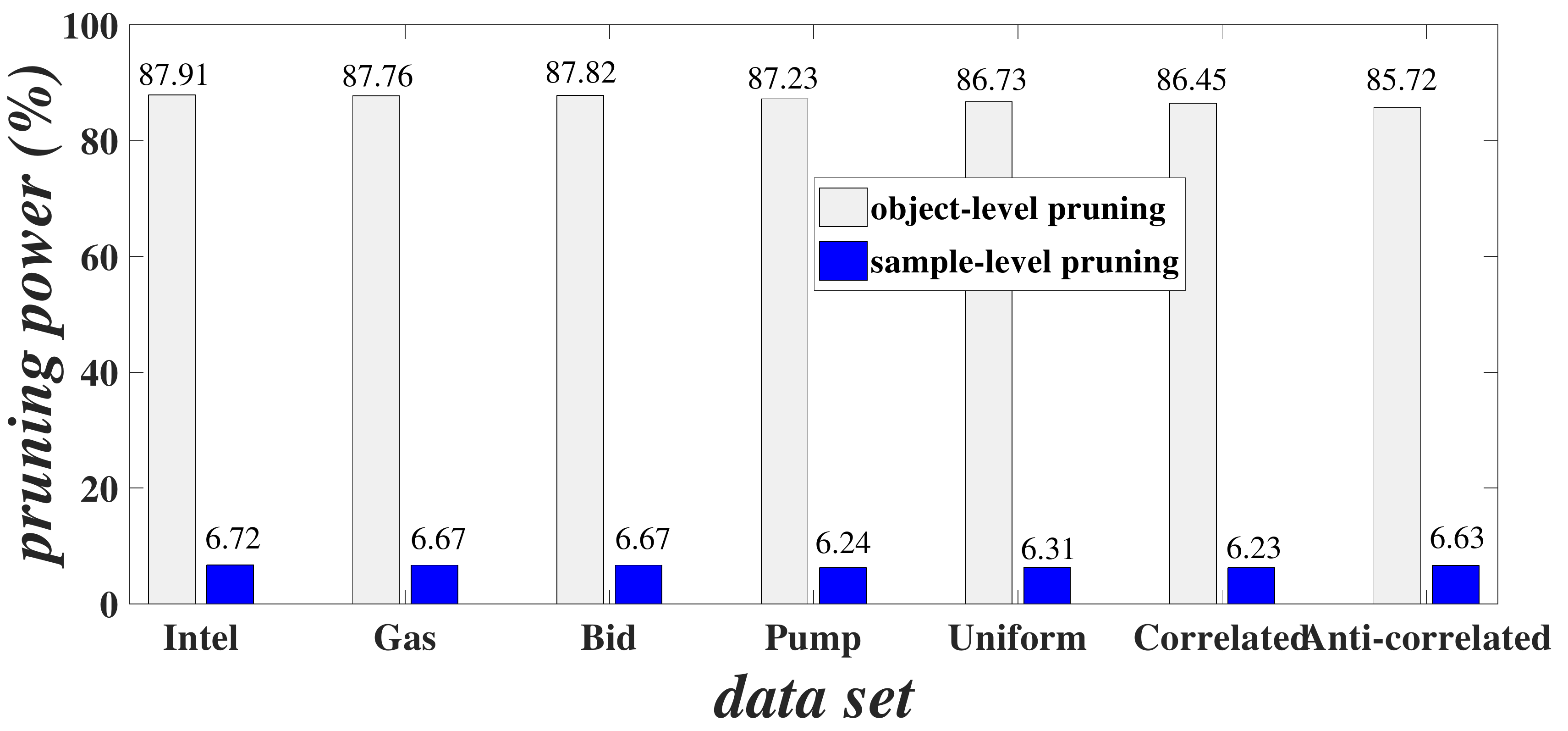}}\vspace{-3ex}
\caption{\small Pruning power evaluation over real/synthetic data sets.}
\label{fig:exper:Join-iDS_pru_pow_dataset}
\end{figure}

\subsection{Effectiveness of Sky-iDS Pruning Methods}
Figure \ref{fig:exper:Join-iDS_pru_pow_dataset} demonstrates the percentages of object pairs that are pruned by our two pruning rules, \textit{object-level pruning} and \textit{sample-level pruning}, over real/synthetic data sets, where parameters are set to their default values in Table \ref{table:exp_parameter_setting}. As mentioned in Section \ref{subsec:pruning_rules}, we will first apply the \textit{object-level pruning}, and then apply the \textit{sample-level pruning} if the former one does not work. From the figure, we can see that the \textit{object-level pruning} can prune most pairs of objects from two different data streams for both real and synthetic data sets (i.e., 87.23\%-87.91\% for real data sets and 85.72\%-86.73\% for synthetic data sets). In addition, the \textit{sample-level pruning} can further prune 6.24\%-6.72\% and 6/23\%-6.63\% object pairs for real and synthetic data sets, respectively. Overall, our proposed pruning rules can together prune 93.47\%-94.63\% and 92.35\%-93.04\% object pairs over real and synthetic data sets, respectively, which confirms the effectiveness of our proposed pruning methods.

\begin{figure}[t!]
\centering 
\subfigure[][{\small $F_1$ score (real data)}]{\hspace{-2ex}
\scalebox{0.2}[0.2]{\includegraphics{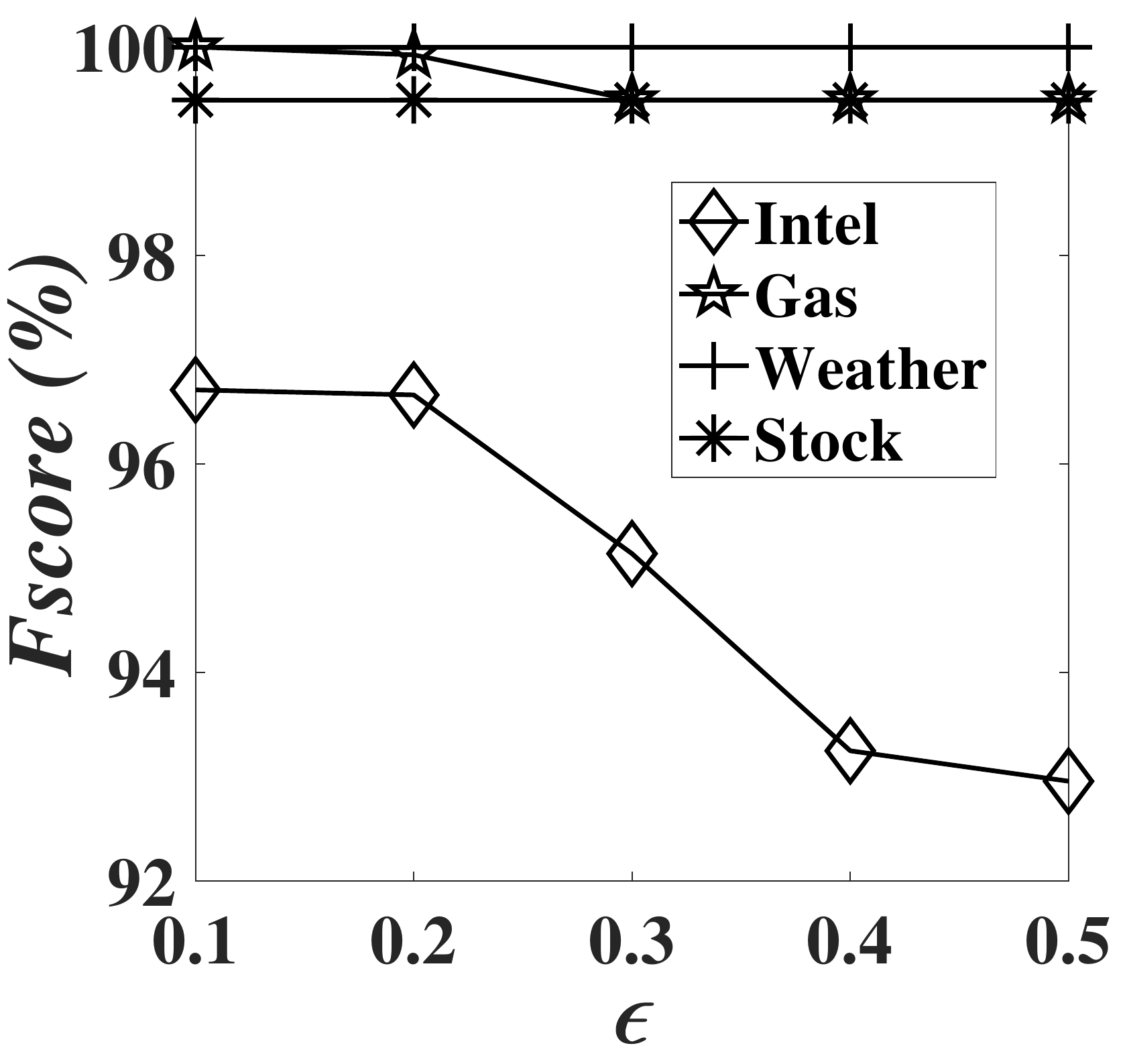}}
\label{subfig:real_Fscore_vs_epsilon}
}\qquad%
\subfigure[][{\small $F_1$ score (synthetic data)}]{
\scalebox{0.2}[0.2]{\includegraphics{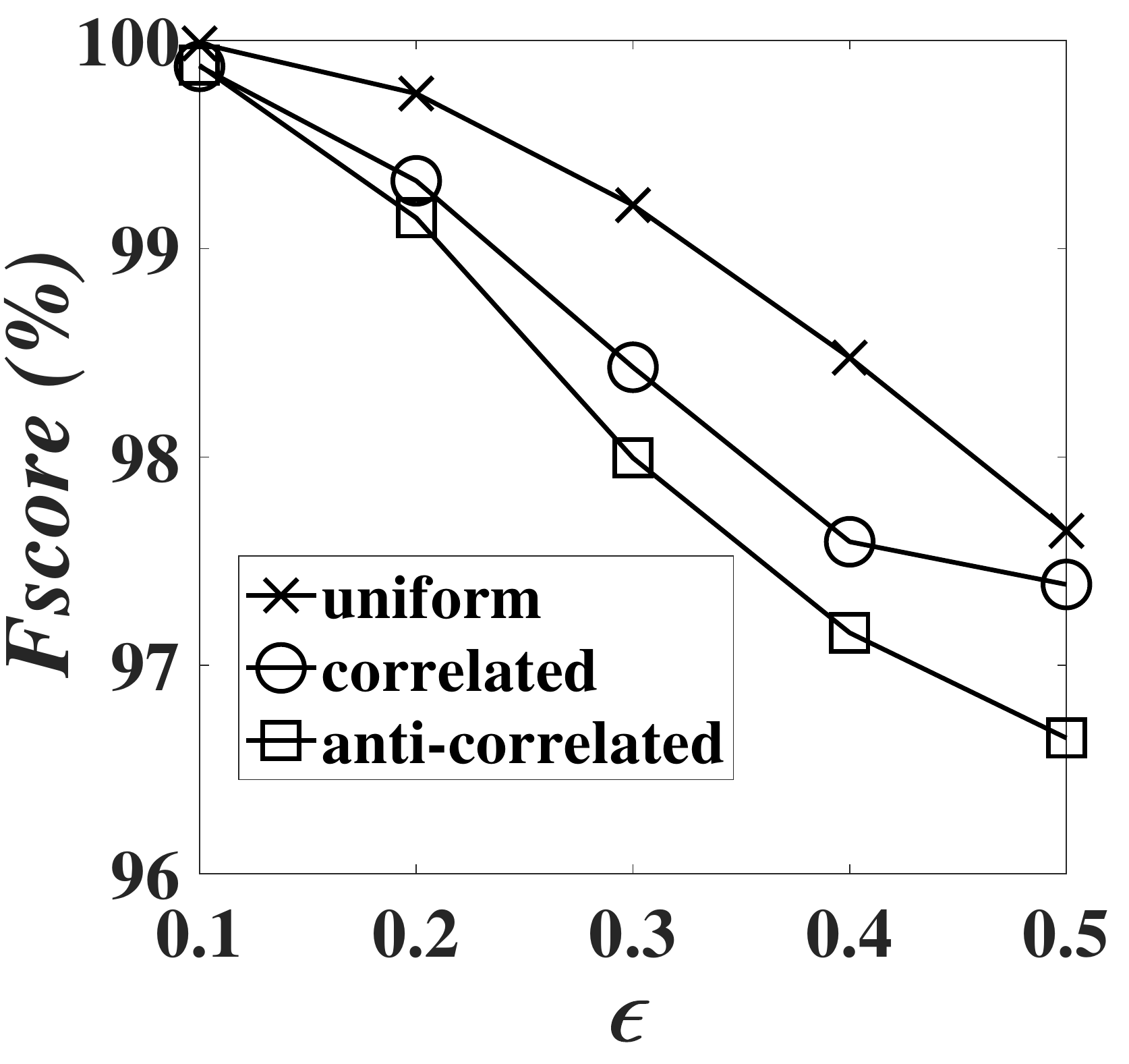}}
\label{subfig:synthetic_Fscore_vs_epsilon}
}\vspace{-3ex}
\caption{\small The Join-iDS effectiveness vs. distance threshold $\epsilon$.} 
\label{exper:Join-iDS_effectiveness_vs_epsilon}
\end{figure} 

\begin{figure}[t!]\vspace{-2ex}
\centering
\scalebox{0.21}[0.21]{\includegraphics{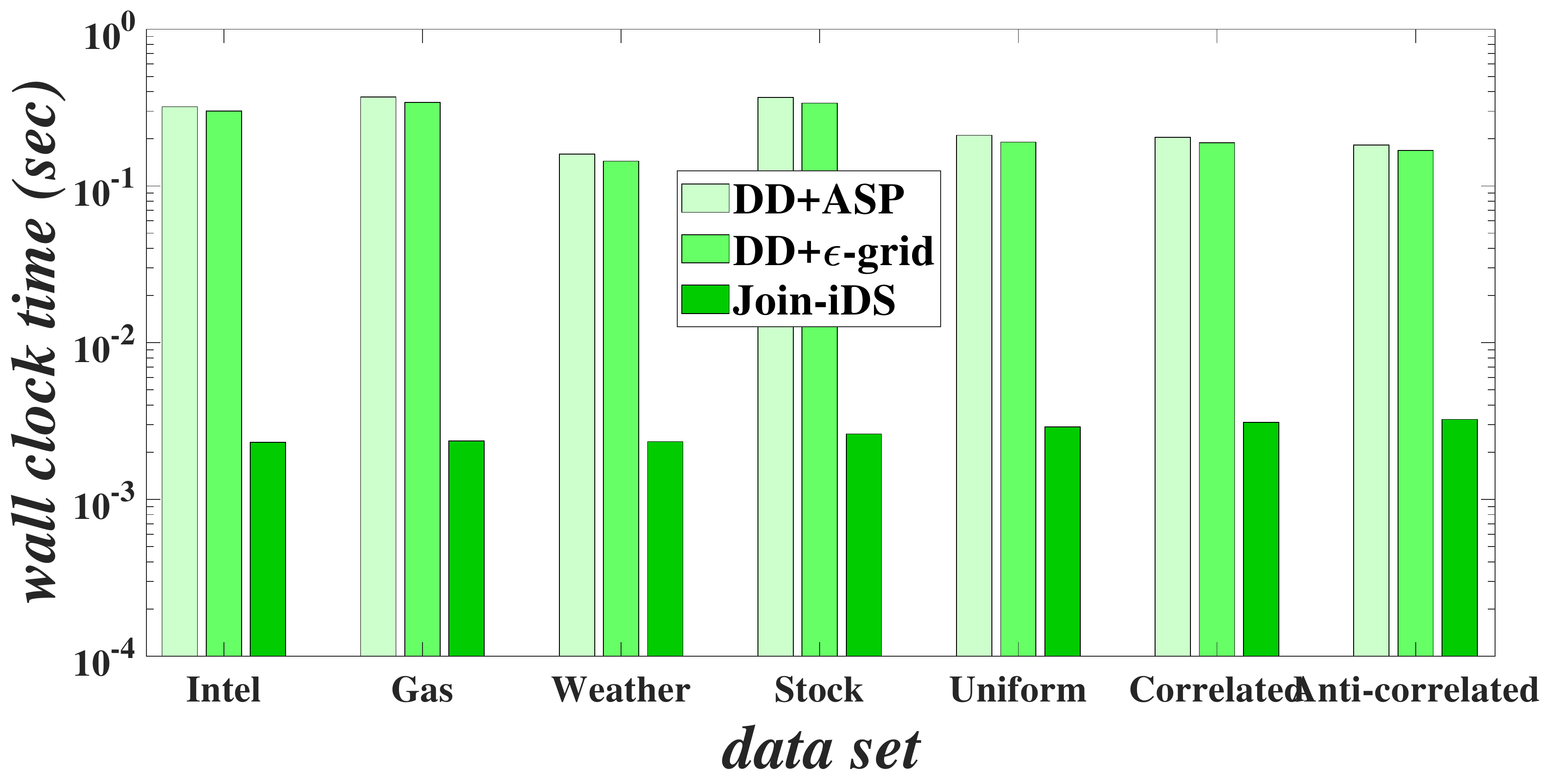}}\vspace{-3ex}
\caption{\small The performance vs. real/synthetic data sets.}
\label{fig:exper:Join-iDS_vs_dataset}
\end{figure}

\subsection{The Effectiveness of Join-iDS}
\label{subsec:Join_effectiveness}

\noindent {\bf The Join-iDS Effectiveness vs. Distance Threshold $\epsilon$.} Figure \ref{exper:Join-iDS_effectiveness_vs_epsilon} illustrates the \textit{$F_1$ score} of our Join-iDS approach over real/synthetic data sets, by varying distance threshold $\epsilon$ from $0.1$ to $0.5$, where other parameters are set to their default values, and \textit{$F_1$ score} is defined as:\vspace{-2ex}

\begin{eqnarray}
F_1\ score &=& 2 \times \frac{recall \times precision}{recall + precision},
\label{eq:F1_score}
\end{eqnarray}

\noindent where \textit{recall} is the number of correct returned join pairs by our Join-iDS approach divided by the number of actual join pairs (i.e., groundtruth); and \textit{precision} is the ratio of correct join pairs among all returned join pairs by our Join-iDS approach. Here, we generate incomplete objects by randomly selecting some attributes in complete data sets as missing, thus, we can know the groundtruth of acutal join results. From experimental results, we can see that the \textit{$F_1$ score} remains high for both real and synthetic data (i.e., above $92\%$ and $96\%$, resp.) for different $\epsilon$ values, which verifies the effectiveness of our imputation and Join-iDS approaches. Note that, the two baseline methods $DD+ASP$ and $DD+\epsilon$\text{-}$grid$ have same \textit{$F_1$ score} as our Join-iDS approach, since they also apply DDs as their imputation methods. Thus, we will not report the effectiveness of the $DD+ASP$ and $DD+\epsilon$\text{-}$grid$. We also tested other parameters, and will not report similar experimental results here.

\subsection{The Efficiency of Join-iDS}
\label{subsec:Join_efficiency}

\noindent {\bf The Join-iDS Performance vs. Real/Synthetic Data Sets.} Figure \ref{fig:exper:Join-iDS_vs_dataset} compares the wall clock time of our Join-iDS approach with that of $DD+ASP$ and $DD+\epsilon$\text{-}$grid$ on real/synthetic data sets, where default parameter values are used (as depicted in Table \ref{table:exp_parameter_setting}). From figures, our Join-iDS approach outperforms the $DD+ASP$ and $DD+\epsilon$\text{-}$grid$ by about 2 orders of magnitude, which confirms the efficiency of the ``data imputation and join processing at the same time'' style of our Join-iDS approach. 


Below, we evaluate the robustness of our Join-iDS approach by varying different parameter values. To clearly illustrate the trend of our Join-iDS approach, we will omit similar results for $DD+ASP$ and $DD+\epsilon$\text{-}$grid$.

\begin{figure}[t!]
\centering\vspace{-1ex}
\subfigure[][{\small real data}]{\hspace{-2ex}
\scalebox{0.2}[0.2]{\includegraphics{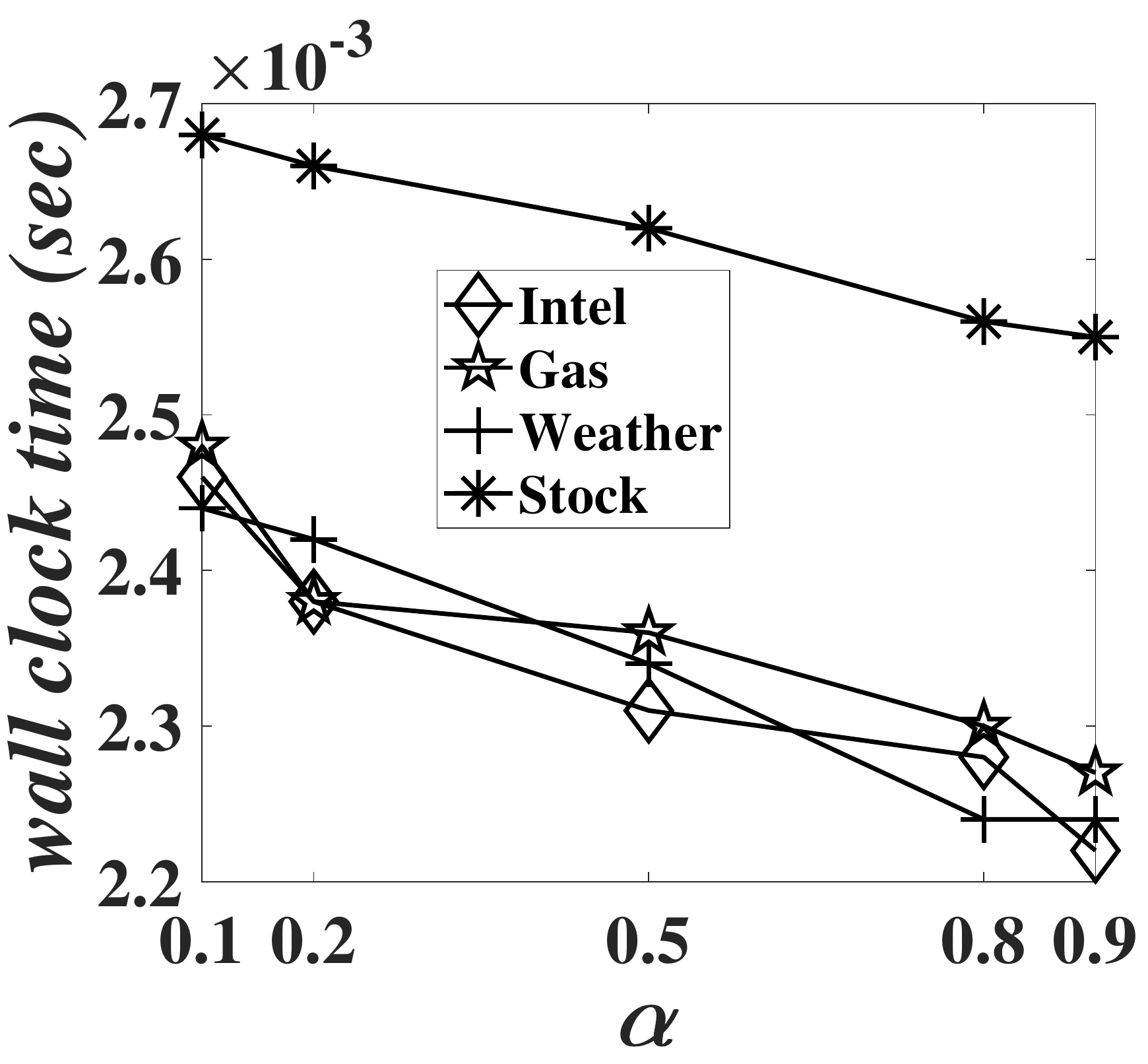}}
\label{subfig:real_cost_vs_alpha}
}\qquad%
\subfigure[][{\small synthetic data}]{\vspace{1ex}
\scalebox{0.2}[0.2]{\includegraphics{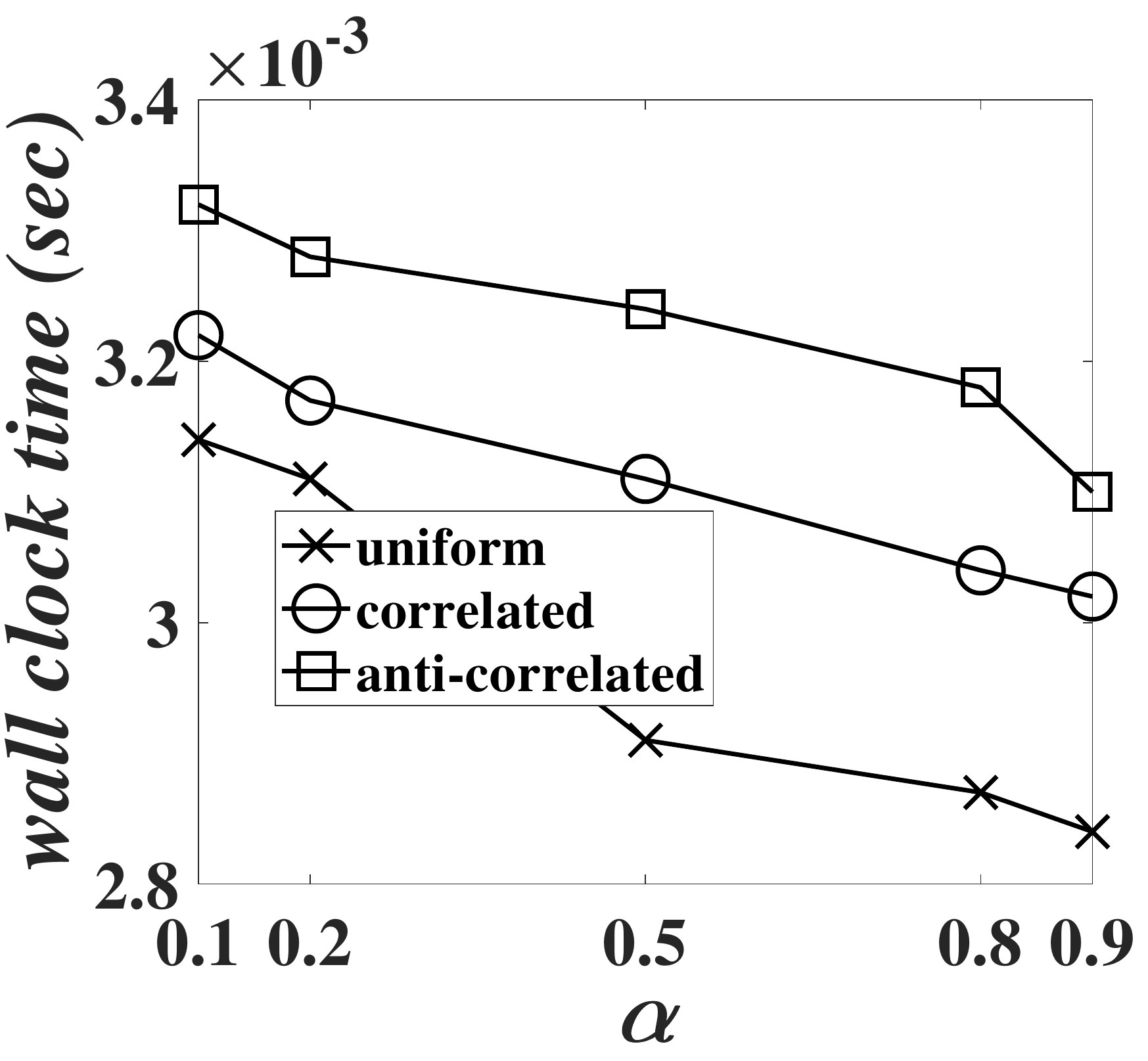}}
\label{subfig:synthetic_cost_vs_alpha}
}\vspace{-3ex}
\caption{\small The performance vs. probabilistic threshold $\alpha$.} 
\label{exper:Join-iDS_vs_alpha} \vspace{-1ex}
\end{figure} 

\begin{figure}[t!]
\centering\vspace{-3ex}
\subfigure[][{\small real data}]{\hspace{-2ex}
\scalebox{0.2}[0.2]{\includegraphics{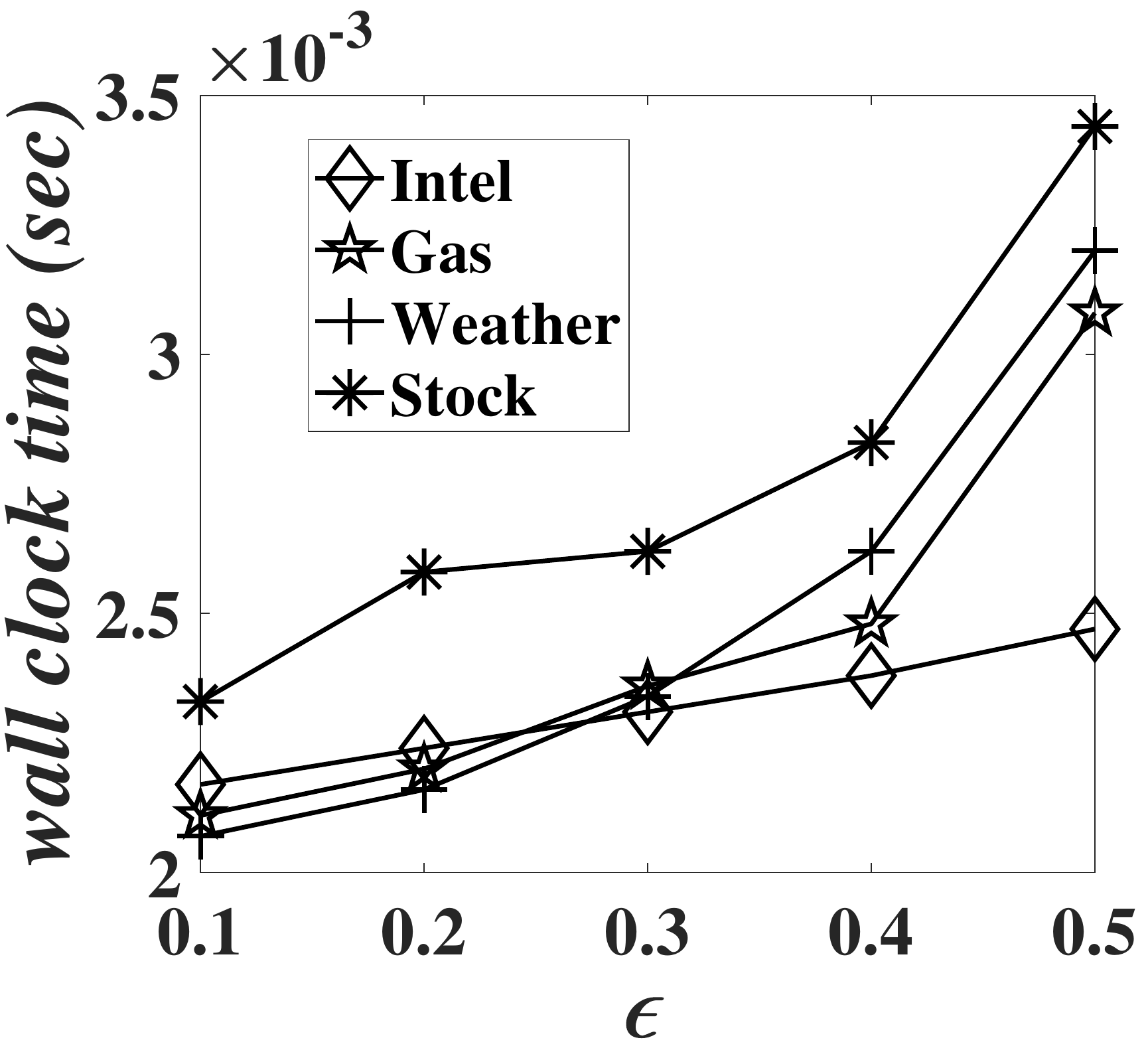}}
\label{subfig:real_cost_vs_epsilon}
}\qquad%
\subfigure[][{\small synthetic data}]{
\scalebox{0.2}[0.2]{\includegraphics{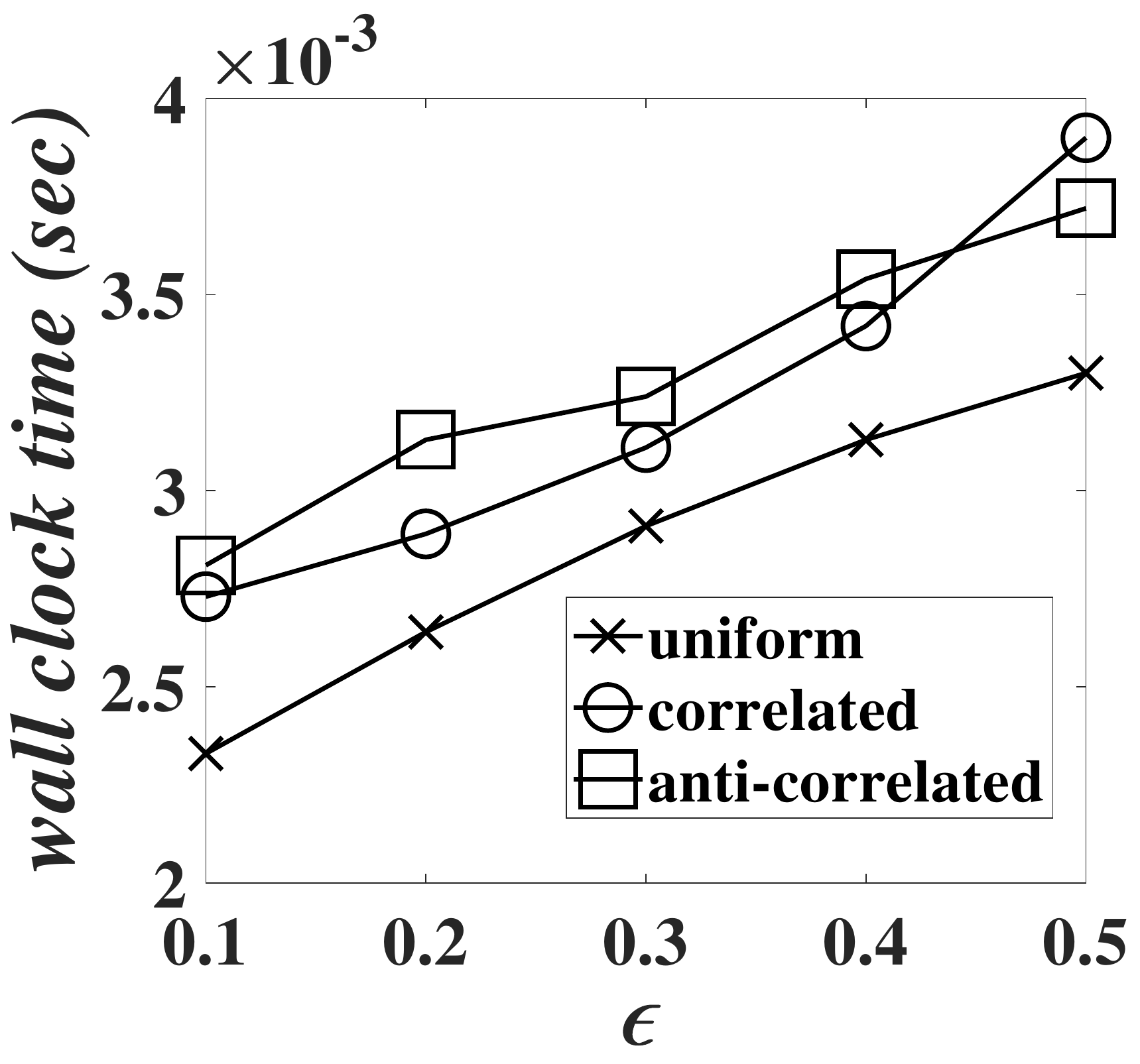}}
\label{subfig:synthetic_cost_vs_epsilon}
}\vspace{-3ex}
\caption{\small The performance vs. parameter $\epsilon$.} 
\label{exper:Join-iDS_vs_epsilon} \vspace{-1ex}
\end{figure} 

\begin{figure}[t!]
\centering\vspace{-3ex}
\subfigure[][{\small real data}]{\hspace{-2ex}
\scalebox{0.21}[0.21]{\includegraphics{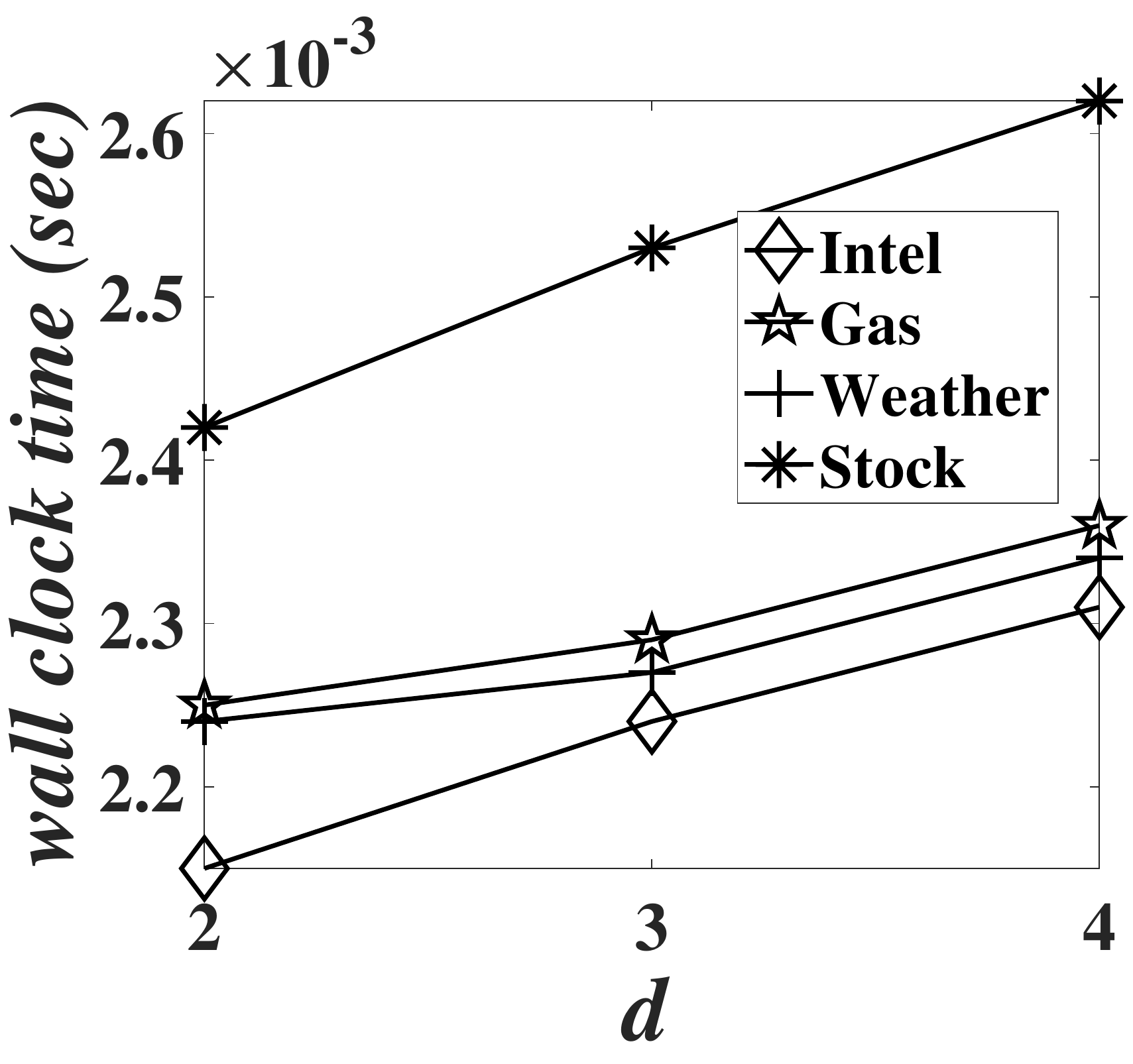}}
\label{subfig:real_cost_vs_d}
}\qquad%
\subfigure[][{\small synthetic data}]{
\scalebox{0.2}[0.2]{\includegraphics{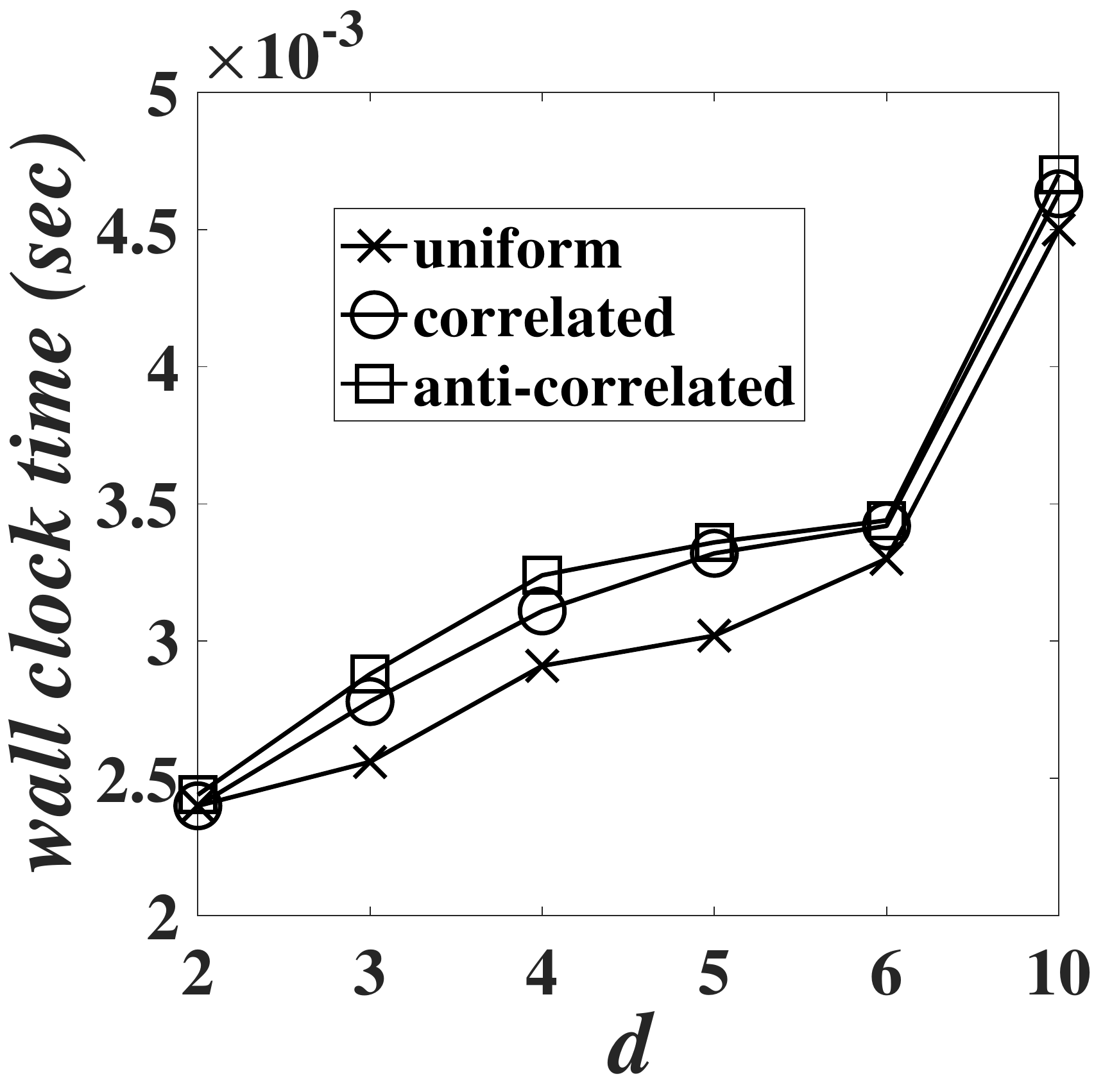}}
\label{subfig:synthetic_cost_vs_d}
}\vspace{-3ex}
\caption{\small The performance vs. dimensionality $d$.} 
\label{exper:Join-iDS_vs_d} 
\end{figure} 

\begin{figure}[t!]
\centering\vspace{-3ex}
\subfigure[][{\small real data}]{\hspace{-2ex}
\scalebox{0.2}[0.2]{\includegraphics{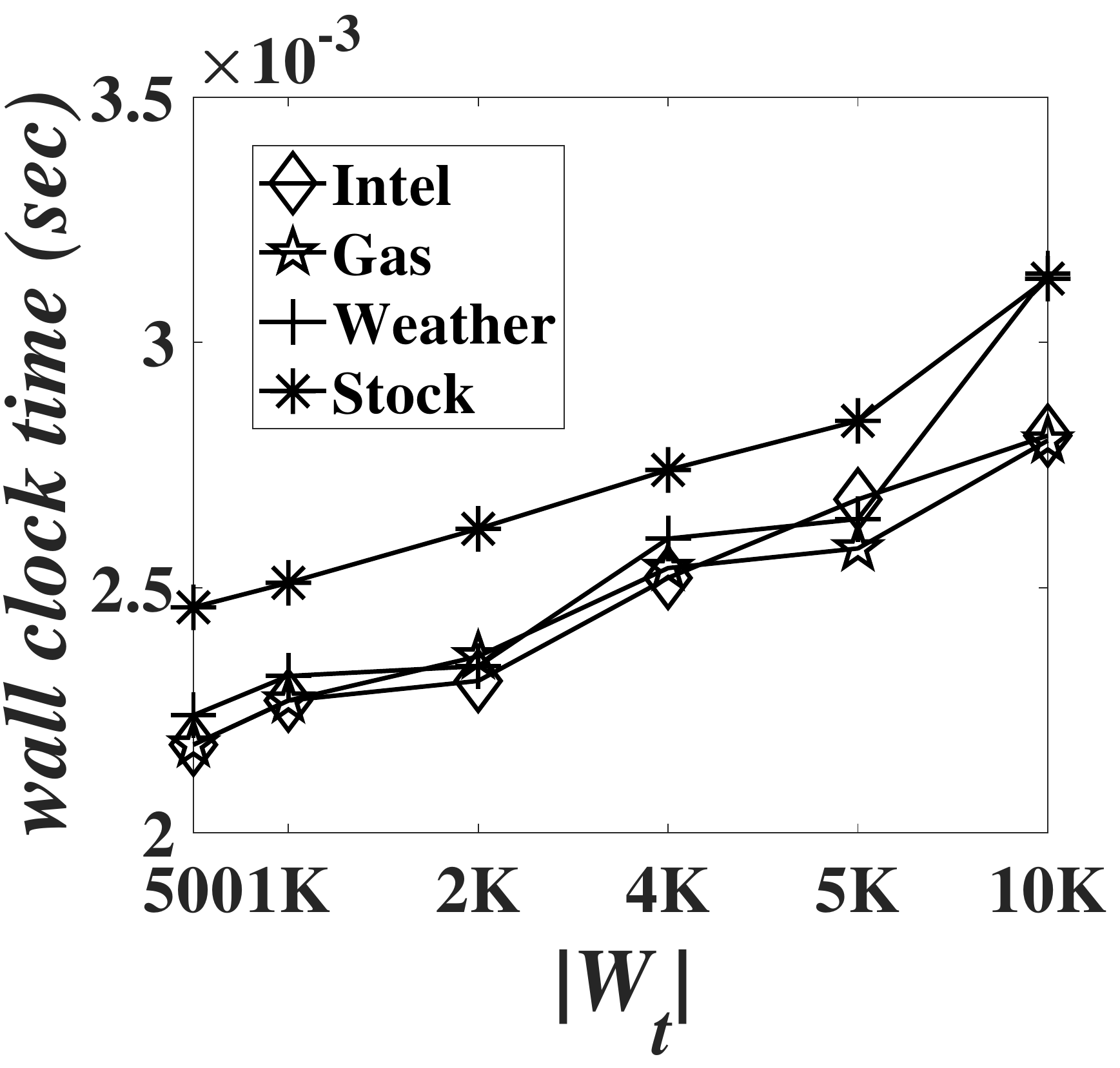}}
\label{subfig:real_cost_vs_w}
}\qquad%
\subfigure[][{\small synthetic data}]{
\scalebox{0.2}[0.2]{\includegraphics{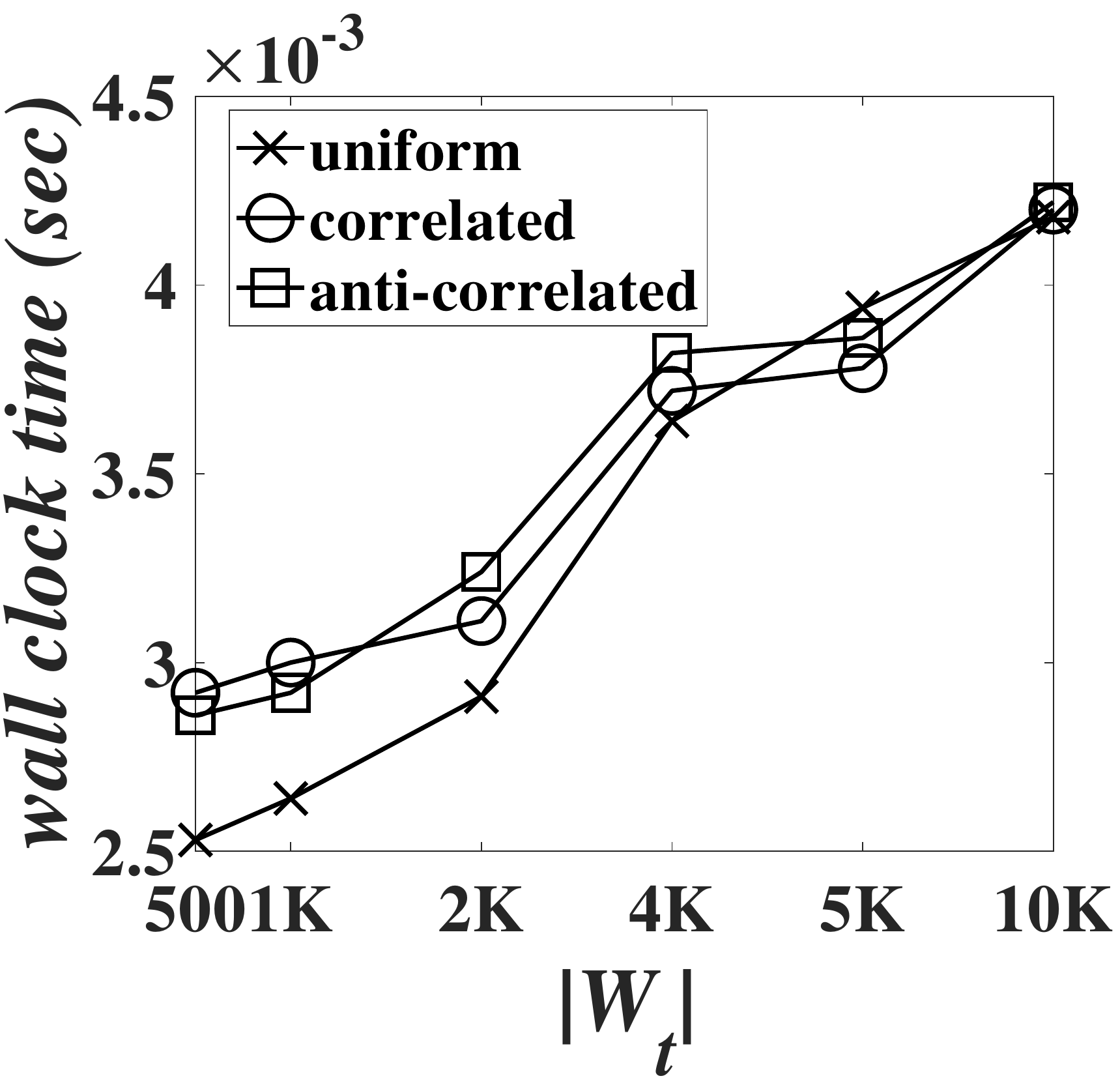}}
\label{subfig:real_cost_vs_w}
}\vspace{-3ex}
\caption{\small The performance vs. No., $|W_t|$, of objects in $W_{1t}$ (or $W_{2t}$).} 
\label{exper:Join-iDS_vs_w} 
\end{figure}

\nop{
\begin{figure}[t!]
\centering \vspace{-3ex}
\subfigure[][{\small real data}]{\hspace{-2ex}
\scalebox{0.2}[0.2]{\includegraphics{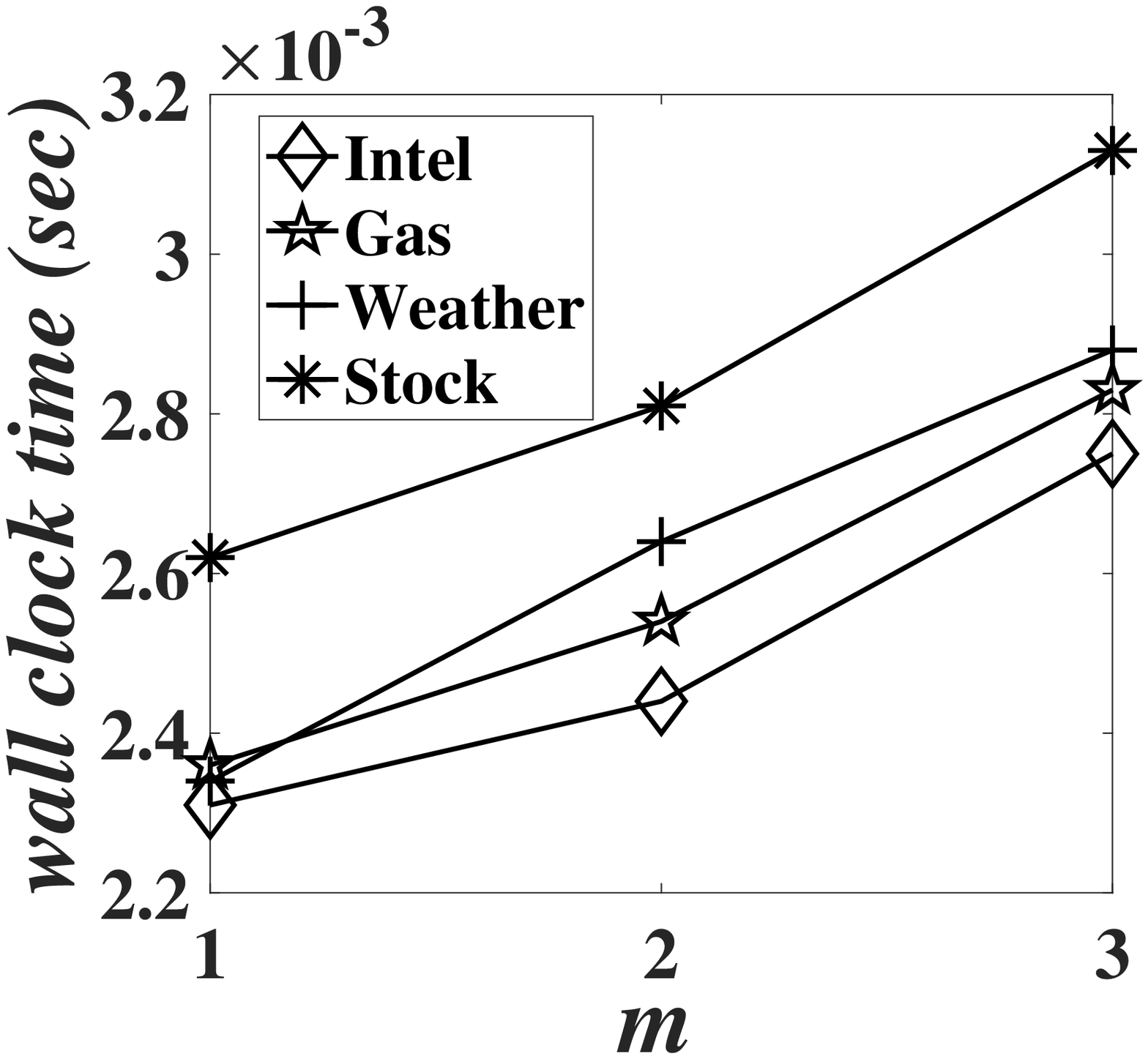}}
\label{subfig:real_cost_vs_m}
}\qquad%
\subfigure[][{\small synthetic data}]{
\scalebox{0.2}[0.2]{\includegraphics{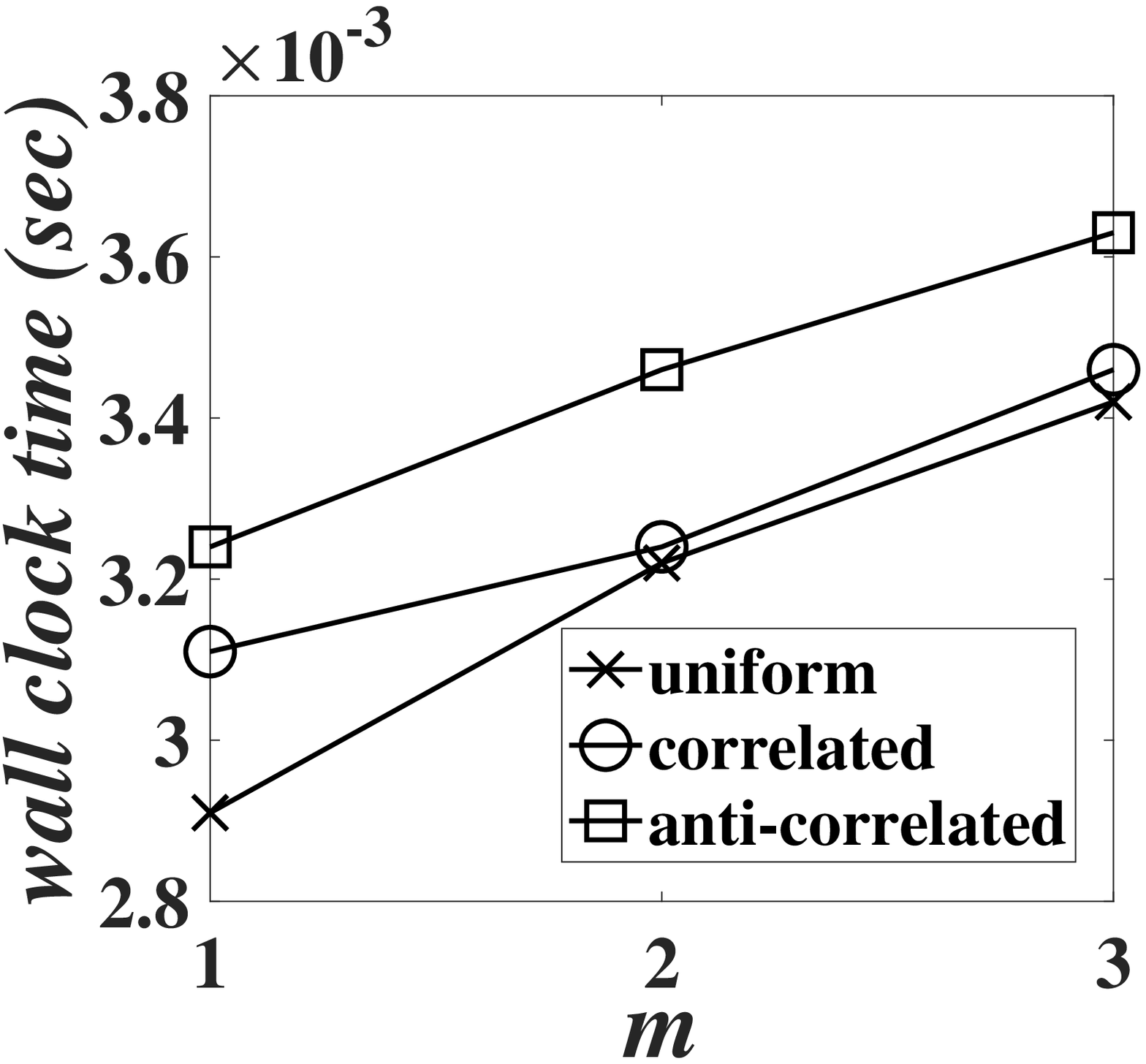}}
\label{subfig:synthetic_cost_vs_R}
}\vspace{-3ex}
\caption{\small The performance vs. No., $m$, of missing attributes.} 
\label{exper:Join-iDS_vs_m}
\end{figure} 
}

\noindent {\bf The Join-iDS Performance vs. Probabilistic Threshold $\alpha$.} Figure \ref{exper:Join-iDS_vs_alpha} shows the performance of our Join-iDS approach, where probabilistic threshold $\alpha=0.1, 0.2, 0.5, 0.8$, and $0.9$, and default values are used for other parameters. From Inequality~(\ref{eq:eq2}), larger $\alpha$ value will incur fewer object pairs that can be joined, and thus lead to smaller wall clock time (as confirmed in figures). For all real/synthetic data, the wall clock time remains low (i.e., less than 0.0027 $sec$ and 0.0034 $sec$, resp.), which indicates the efficiency of our Join-iDS approach with different $\alpha$ values.

\noindent {\bf The Join-iDS Performance vs. Distance Threshold $\epsilon$.} Figure \ref{exper:Join-iDS_vs_epsilon} evaluates the effect of distance threshold $\epsilon$ on our Join-iDS performance, where $\epsilon$ varies from $0.1$ to $0.5$, and other parameter values are by default. Intuitively, larger $\epsilon$ values incur lower pruning power. Thus, when $\epsilon$ becomes larger, the wall clock time smoothly increases. Nonetheless, the wall clock time still remains low for real and synthetic data (i.e., less than 0.0035 $sec$ and 0.0039 $sec$, resp.).


\noindent {\bf The Join-iDS Performance vs. Dimensionality $d$.} Figure \ref{exper:Join-iDS_vs_d} varies the dimensionality, $d$, of objects in streams from $2$ to $4$ for real data, and from $2$ to $10$ for synthetic data, where other parameters are set to their default values. As the increase of $d$, all real/synthetic data sets need more wall clock time, which is due to ``the curse of dimensionality'' problem \cite{Berchtold96}. Nevertheless, wall clock times for all real/synthetic data still remain low (i.e., below 0.0026 $sec$ and 0.0049 $sec$, resp.), which verifies good Join-iDS performance.


\noindent {\bf The Join-iDS Performance vs. the Number, $|W_t|$, of Objects in Sliding Windows.} Figure \ref{exper:Join-iDS_vs_w} illustrates the performance of our Join-iDS approach for different sizes, $|W_t|$, of sliding windows (i.e., $W_{1t}$ and $W_{2t}$), where $|W_t|=500$, $1K$, $2K$, $4K$, $5K$, and $10K$, and other parameters are set to their default values. In figures, the wall clock time increases for larger $|W_t|$. This is reasonable, since we need to maintain more objects in $\epsilon$-grid and check more candidate join pairs in $JS$. Nonetheless, the wall clock time remains low for real/synthetic data (i.e., less than 0.0032 $sec$ and 0.0043 $sec$, resp.), which shows good scalability of our Join-iDS approach for large window size.

We also tested other parameters (e.g., the size, $|R|$, of the data repository $R$ and the number, $m$, of missing attributes), and will not report similar experimental results here due to space limitations. In summary, our Join-iDS approach can achieve good performance under different parameter settings.

\section{Related Work}
\label{sec:related_work}

\noindent {\bf Stream Processing.} There are many important problems for stream data processing, including event detection \cite{poppe2017complete}, outlier detection \cite{sadik2018wadjet}, top-$k$ query  \cite{mouratidis2006continuous}, join \cite{das2003approximate,lian2009efficient}, skyline query \cite{tao2006maintaining}, nearest neighbor query \cite{bohm2007efficiently}, aggregate query \cite{tatbul2006window}, and so on. These works usually assume that stream data are either certain or uncertain. To our best knowledge, they cannot be directly applied to our Join-iDS problem, under the semantics of incomplete data streams.

\nop{
\noindent {\bf Data Imputation in Data streams.}
{\color{Weilong} do we need to add this part?}
}

\noindent {\bf Differential Dependency.} Differential dependency (DD) \cite{song2011differential} is a valuable tool for data imputation \cite{song2015enriching}, data cleaning \cite{prokoshyna2015combining}, data repairing \cite{song2014repairing}, and so on. Song et al. \cite{song2015enriching} used the DDs to fill the missing attributes of incomplete objects on static data set via some detected neighbors satisfying the distance constraints on determinant attributes. Song et al. \cite{song2014repairing,song2017graph} also explored to repair labels of graph nodes. Prokoshyna et al. \cite{prokoshyna2015combining} cleaned databases by removing inconsistent records that violate DDs. Unlike these works targeting at static database, we apply DD-based imputation to the streaming environment, which makes our Join-iDS problem more challenging.



\noindent {\bf Join Over Certain/Uncertain Databases.} The join operator was traditionally used in relational databases \cite{mishra1992join} or data streams \cite{das2003approximate}. The join predicate may follow equality semantics between attributes of tuples or data objects. According to predicate constraints, join over uncertain databases \cite{lian2010similarity,lian2009efficient} can be classified into two categories, probabilistic join query (PJQ) and probabilistic similarity join (PSJ), which return pairs of joining objects that are identical or similar (e.g., within $\epsilon$-distance from each other), resp., with high confidences. 

PSJ has received much attention in many domains. Galkin et al. \cite{galkin2017sjoin} applied PSJ to integrate heterogeneous RDF graphs by introducing an equivalent semantics for RDF graphs. Ma et al. \cite{ma2017novel} proposed an effective filter-based method for high-dimensional vector similarity join. Wang et al. \cite{wang2017leveraging} explored how to leverage relations between sets to proceed the exact set similarity. Li et al. \cite{li2015efficient} proposed a prefix tree index to join multi-attribute Data. Shang et al. \cite{shang2017trajectory} applied PSJ in trajectory similarity join in spatial networks via some search space pruning techniques. Bohm et al. \cite{bohm2001epsilon} proposed a join approach for massive high-dimensional data, based on a particular order of data points via a grid. Different from \cite{bohm2001epsilon} that uses the grid cell for sorting data points, in our work, we designed a grid variant, $\epsilon$-grid, which stores additional information (e.g., queues with imputed objects) specific for incomplete data streams, and supports the dynamic maintenance of candidate join answers. 

Existing works on join over certain or uncertain databases (or data streams) usually assume that, the underlying data have complete attributes. Thus, their techniques cannot be directly applied to solve our Join-iDS problem in the presence of missing attributes. 


\noindent {\bf Incomplete Databases.} In the literature of incomplete databases, the most commonly used imputation methods include rule-based \cite{fan2010towards}, statistical-based \cite{mayfield2010eracer}, pattern-based \cite{wellenzohn2017continuous}, constraint-based \cite{zhang2017time} imputation, and so on. These existing works may incur the accuracy problem for sparse data sets. That is, sometimes, they may not be able to find samples to impute the missing attributes in sparse data sets, which may lead to problems such as imputation failure or even wrong imputation result \cite{song2011differential}. To avoid or alleviate this problem, in this paper, we use DDs to impute missing attributes based on a historical (complete) data repository $R$.  We will consider the regression-based imputation approaches (e.g., \cite{8731351}) as our future work.

\section{Conclusions}
\label{sec:conclusions}

In this paper, we formalize the problem of the \textit{join over incomplete data streams} (Join-iDS), which is useful for many real applications such as sensor data monitoring and network intrusion detection. In order to tackle the Join-iDS problem, we design a cost-model-based data imputation method via DDs, devise effective pruning and indexing mechanisms, and propose an efficient algorithm to incrementally maintain the join results over the imputed data streams. Through extensive experiments, we confirm the efficiency and effectiveness of our proposed Join-iDS approach on both real and synthetic data.

\balance
%

\begin{acks}
Xiang Lian is supported by NSF OAC No. 1739491 and Lian Start
up No. 220981, Kent State University. We thank the anonymous reviewers for the useful suggestions.
\end{acks}

%
\bibliographystyle{ACM-Reference-Format}
\bibliography{sample-base}

%

\newpage
\appendix

\text{ }\\
\noindent {\huge \bf Appendix}

\section{Proof for Pruning Strategies}
\label{sec:proof}

\subsection{Proof of Lemma \ref{lemma:lem1}}
\label{subsec:proof_lem1}
\begin{proof}
If $mindist(o_x^p,o_y^p) > \epsilon$ holds, we can get $Pr\{mindist(o_x^p, o_y^p)$ $\leq \epsilon\}=0$ based on Definition \ref{def:Join-iDS}. Since $mindist(o_x^p,o_y^p)$ is the minimum distance between imputed objects $o_x^p$ and $o_y^p$, we can obtain $Pr_{Join\text{-}iDS}(o_x^p,$ $o_y^p)$ $= Pr\{dist(o_x^p, o_y^p) \leq \epsilon\} \le Pr\{mindist(o_x^p, o_y^p) \leq \epsilon\} = 0$. That is, $Pr_{Join\text{-}iDS}(o_x^p,o_y^p)=0$.
\end{proof} 

\nop{
\subsection{Proof of Lemma \ref{lemma:lem2}}
\label{subsec:proof_lem2}
\begin{proof}
Given three data points, $c_x$, $c_y$ and $pvt$, in data space, based on triangle inequality, we can get $dist(c_x,c_y) > |dist(c_x,pvt)-dist(c_y,pvt)|$.
In addition, based on the space position of the MBRs of objects $o_x^p$ and $o_y^p$, we can get $dist(c_x,c_y) \le dist(o_x^p,o_y^p) +r_x +r_y$. Thus, we can get $Pr_{Join\text{-}iDS}(o_x^p,o_y^p) = Pr\{dist(o_x^p, o_y^p) \le \epsilon\} \le Pr\{dist(c_x, c_y) - r_x - r_y \le \epsilon\} < Pr\{|dist(c_x,pvt)-dist(c_y,pvt)| - r_x - r_y \le \epsilon\}$. If $|dist(c_x,pvt)-dist(c_x,pvt)|>\epsilon + r_x + r_y$, we can get $Pr\{|dist(c_x,pvt)-dist(c_y,pvt)| - r_x - r_y \le \epsilon\} = 0$. Finally we can get $Pr_{Join\text{-}iDS}(o_x^p,o_y^p) < Pr\{|dist(c_x,pvt)-dist(c_y,pvt)| - r_x - r_y \le \epsilon\} = 0$.
\end{proof} 
}

\subsection{Proof of Lemma \ref{lemma:lem3}}
\label{subsec:proof_lem3}
\begin{proof}
Given the two sub sets (MBRs) $s_x \in o_x^p.MBR$ and $s_y \in o_y^p.MBR$, based on \textit{Lemma} \ref{lemma:lem1}, if their minimum distance is larger than $\epsilon$ (i.e., $mindist(s_x,s_y) > \epsilon$), we can get $s_x$ and $s_y$ cannot be jointed with confidence 1 (i.e., $Pr_{Join\text{-}iDS}(s_x,s_y) = 0$). We use $Pr\{o_{xl} \in s_x \land o_{yg} \in s_y\}$ to represent the probability that instances $o_{xl}$ and $o_{yg}$ are from MBRs $s_x$ and $s_y$, respectively, while we use $Pr\{o_{xl} \notin s_x \lor o_{yg} \notin s_y\}$ to indicate the probability that $o_{xl}$ and $o_{yg}$ are not all from $s_x$ and $s_y$.  Thus, based on conditional probability and Eq. (\ref{eq:eq3}), we can get $Pr_{Join\text{-}iDS}(o_x,o_y) = \sum_{o_{xl}} \sum_{o_{yg}} Pr\{o_{xl} \in s_x \land o_{yg} \in s_y\} \cdot o_{xl}.p \cdot o_{yg}.p \cdot \chi(dist(o_{xl},o_{yg})\le \epsilon) + \sum_{o_{xl}} \sum_{o_{yg}} Pr\{o_{xl} \notin s_x \lor o_{yg} \notin s_y\} \cdot o_{xl}.p \cdot o_{yg}.p  \cdot \chi(dist(o_{xl},o_{yg})\le \epsilon) \le \sum_{o_{xl}} \sum_{o_{yg}} Pr\{o_{xl} \notin s_x \lor o_{yg} \notin s_y\} \cdot o_{xl}.p \cdot o_{yg}.p  \cdot \chi(dist(o_{xl},o_{yg})\le \epsilon) \le \sum_{o_{xl}} \sum_{o_{yg}} Pr\{o_{xl} \notin s_x \lor o_{yg} \notin s_y\} = 1 - \sum_{o_{xl}} \sum_{o_{yg}} Pr\{o_{xl} \in s_x \land o_{yg} \in s_y\} = 1 - \beta_x \cdot \beta_y$. Since $\beta_x \cdot \beta_y > 1-\alpha$, thus, finally we can get $Pr_{Join\text{-}iDS}(o_x,o_y) \le 1 -  \beta_x \cdot \beta_y < 1 - (1 - \alpha) = \alpha$.
\end{proof} 

\section{Cost Models}
\label{sec:cost_model}

\subsection{Cost Model for the DD Selection via Fractal Dimension}
\label{subsec:cost_DD_selection}
Given a DD rule, $Y\to A_j$, returned by \textit{imputation lattice} $Lat_j$ and an incomplete object $o_i$ with missing attribute $A_j$, we will search samples $o_c\in R$ via a query range, $Q$, enclosed by the intervals, $A_y.I$, on attributes $A_y\in Y$. Thus, we can project data points in $R$ from $d$ dimension to $|Y|$ dimensions, where $|Y|$ is the number of attributes in $Y$. Within the $|Y|$-dimensional space, we can calculate its \textit{fractal dimension} \cite{belussi1998self}, denoted as $D_2^Y$, as follows.
\begin{eqnarray}
D_2^Y = \frac{\partial log(\sum p_i^2)}{\partial log(r)}, \ \ \ \ \ \ r \in (r_1, r_2)
\label{eq:D2_calculation}
\end{eqnarray}
\noindent where the reduced $|Y|$-dimensional space of data repository $R$ is composed of multiple regular cells with equal side length $r \in (r_1, r_2)$, in which contains $p_i$ percentage of data points in $R$.

With the \textit{fractal dimension}, $D_2^Y$, of the projected $|Y|$-dimensional data space, we can obtain the estimated count, $cnt_Q$, of objects $o_c\in R$ falling into the query range $Q$, via Eq. (\ref{eq:candidate_num_estimation}) \cite{belussi1998self}:
\begin{eqnarray}
cnt_Q = \left(\frac{Vol(\epsilon, Q)}{Vol(\epsilon, \square)}\right)^{\frac{D_2^Y}{|Y|}} \times (N_{\square}-1) \times 2^{D_2^Y} \times \epsilon^{D_2^Y},
\label{eq:candidate_num_estimation}
\end{eqnarray}
\noindent where $\epsilon$ is the largest distance between the center of query range $Q$ and corresponding the most remote point on $Q$, $Vol(\epsilon,Q)$ is the volume of $Q$, $Vol(\epsilon,\square)$ is the volume of a $|Y|$-dimension regular cube with the largest distance $\epsilon$ between the most remote point on $\square$ and its center of $\square$, and $N_{\square}$ is the number of objects within the range of $\square$.

\nop{
\subsection{Cost Model for Selecting a Good Pivot Set $PVT^*$ for \textit{Lemma} \ref{lemma:lem2}}
\label{subsec:cost_for_PVT}
As depicted in \textit{Lemma} \ref{lemma:lem2}, for two incomplete objects $o_x$ and $o_y$, a good pivot $pvt \in PVT$ may help prune the false join pair $\{o_x,o_y\}$, without the direct comparison of their distance. Given a number of $h$ pivot sets (can be generated randomly in data space of $R$), it is important to select a good pivot set $PVT^*$ (from $R$) with strong pruning power to prune the most number of object pairs. To obtain such a pivot set $PVT^*$, we devise a specific cost model as follow.

\begin{eqnarray}
PVT^* = arg \max\limits_{PVT} \sum_{\forall o_x}^{N_1} \sum_{\forall o_y}^{N_2} (1- \prod_{\forall pvt \in PVT} Pr\{dist(o_x, o_y) \le \epsilon| pvt\})
\label{eq:eqx}
\end{eqnarray}
\noindent where incomplete objects $o_x$ and $o_y$ are samples (following the distribution of historical data), $N_1$ and $N_2$ are the corresponding number of samples $o_x$ and $o_y$, $Pr\{dist(o_x, o_y) \le \epsilon| pvt\}$ is the probability that the pair, $\{o_x,o_y\}$, of samples $o_x$ and $o_y$ cannot be pruned, based on pivot $pvt$ via \textit{Lemma} \ref{lemma:lem2}. 

Thus we get the probability, $1 - (\prod_{\forall pvt \in PVT} - Pr\{dist(o_x, o_y) \le \epsilon| pvt\})$, that the pair $\{o_x,o_y\}$ can be pruned via any ($h$) pivot $pvt \in PVT$. Finally, we sum up the probabilities of these $N_1 \times N_2$ sample pairs, and select the pivot set $PVT^*$ that maximises the probability sum.
}

\subsection{Cost Model for the Cluster Selection}
\label{subsec:cost_for_cluster}

Given a an incomplete object $o_i$ (with missing attribute $A_j$) and a DD rule $(X\to A_j, \phi[X\ A_j])$, we can obtain a query range in the form of 

\begin{eqnarray}
Q (o_i, DD) =\bigwedge_{A_x \in X} [o_i[A_x] - \epsilon_{A_x}, o_i[A_x] - \epsilon_{A_x}].
\label{eq:query_range}
\end{eqnarray}

With the query range $Q$, we can obtain all intersected clusters $cls_k$ (for $1\le k\le n$) from the cluster set $CLS$. For each intersected cluster $cls_k$, we use $itsec$ to represent the intersection between $cls_k$ and $Q$. Our goal is to find the best cluster set $CLS^*$ that can maximise the ratio of intersection between $Q$ and cluster set $CLS$, that is:
\begin{eqnarray}
CLS^* = arg \max\limits_{CLS} \sum_{i=1}^s \sum_{\forall DD} \sum_{\forall cls_k \in CLS} \frac{itset.N}{cls_k.N} \cdot \omega(cls_k,Q(o_i,DD)),
\label{eq:cost_model}
\end{eqnarray}
\noindent where $s$ is the number of samples, $x.N$ ($x=itset$ or $cls_k$) is the number of samples $o_c$ in $x$, and function $\omega(cls_k,Q(o_i,DD)) = 1$ if the query range $\omega(cls_k,Q(o_i,DD))$ is intersected with cluster $cls_k$ (otherwise $\omega(cls_k,Q(o_i,$ $DD)) = 0$).



\section{Selection of sub-MBRs $s_x$ and $s_y$}
\label{sec:Ij_for_sample_pruning}
Given two objects $o_x^p$ and $o_y^p$ imputed via index $I_j$, to apply the \textit{sample-level pruning}, we need to select two sets (MBRs), $s_x$ and $s_y$, of buckets $buc_f$ from nodes in R$^*$-tree intersected with $o_x^p$ and $o_y^p$. There are exponential number (i.e., $\lambda^2$) of selection combinations between the set pair $(s_x,s_y)$. In this paper, we design an effective selection strategy to select $s_x$ and $s_y$ as below. The general idea is that we first select a $s_x$, based on which we select a $s_y$. To be specific, we will obtain $s_x$ by combining consecutive buckets from $buc_1$ till $buc_{f}$ such that its overall frequency is beyond $(1-\alpha) \cdot e.cnt$, that is, $\beta_x=\frac{s_x.cnt}{e.cnt} = \frac{\sum_{i=1}^f buc_f.cnt}{e.cnt} > 1-\alpha$. Next, we will decide whether to add the next bucket $buc_{f+1}$ to $s_x$ by checking the value $\frac{\Delta S_x.I}{\Delta \beta_x}$, where $\Delta S_x.I$ and $\Delta \beta_x$ are the change ratios between the changes of intervals $s_x.I$ and $\beta_x$ w.r.t. the addition of bucket $buc_{f+1}$ into $s_x$, respectively. If $\frac{\Delta S_x.I}{\Delta \beta_x} < 1$, we will add $buc_{f+1}$ to $s_x$, since this addition will bring to $s_x$ more samples $o_c \in e$ but do not significantly enlarge the interval $s_x.I$ of $s_x$ on attribute $A_j$. Otherwise (i.e., $\frac{\Delta S_x.I}{\Delta \beta_x} \ge 1$), we will not add $buc_{f+1}$ to $s_x$. After we fix the $s_x$, we will select the first $s_y$ with $\beta_y > \frac{1-\alpha}{\beta_x}$. In this case, we can use the selected $s_x$ and $s_y$ to apply Lemma \ref{lemma:lem3} to prune the object pair $(o_x, o_y)$.

\end{document}